\documentclass[11pt]{article}

\usepackage[T5]{fontenc}
\usepackage[margin=1in]{geometry}
\usepackage[labelfont=bf]{caption}
\usepackage{tabularray}

\usepackage{algorithm,algpseudocode,amsmath,amssymb,amsthm,float,mathtools,thmtools,thm-restate,scalerel,stackengine,verbatim,subcaption,bbm}
\usepackage[shortlabels]{enumitem}
\usepackage{tikz}
\usepackage[pagebackref]{hyperref}
\hypersetup{
    colorlinks,
    linkcolor={red!100!black},
    citecolor={blue!100!black},
}
\usepackage[capitalize,nameinlink]{cleveref}
\usepackage{lmodern}
\usetikzlibrary{decorations.pathreplacing, positioning}
 
\stackMath
\newcommand\reallywidehat[1]{%
\savestack{\tmpbox}{\stretchto{%
  \scaleto{%
      \scalerel*[\widthof{\ensuremath{#1}}]{\kern-.6pt\bigwedge\kern-.6pt}%
          {\rule[-\textheight/2]{1ex}{\textheight}}
            }{\textheight}%
            }{0.5ex}}%
            \stackon[1pt]{#1}{\tmpbox}%
            }
\newcommand{\bool}{\set{0,1}}
\newcommand{\bools}{\set{1,-1}}

\newcommand{\Cbb}{\mathbb{C}}

\newcommand{\Nbb}{\mathbb{N}}
\newcommand{\Rbb}{\mathbb{R}}

\newcommand{\Ac}{\mathcal{A}}

\newcommand{\Ec}{\mathcal{E}}
\newcommand{\Fc}{\mathcal{F}}

\newcommand{\Nc}{\mathcal{N}}

\newcommand{\Qc}{\mathcal{Q}}

\newcommand{\Tc}{\mathcal{T}}
\newcommand{\Uc}{\mathcal{U}}

\newcommand{\wh}[1]{\widehat{#1}}
\newcommand{\eps}{\varepsilon}

\newcommand{\prob}[2][]{\operatorname*{\mathbb{P}}_{#1 }\brac*{#2}}
\newcommand{\E}[2][]{\operatorname*{\mathbb{E}}_{#1 }\brac*{#2}}
\newcommand{\Var}[2][]{\operatorname*{\normalfont{\text{Var}}}_{#1 }\paren*{#2}}

\newcommand{\bO}[1]{\operatorname*{O}\paren*{#1}}

\newcommand{\lO}[1]{\operatorname*{o}\paren*{#1}}
\newcommand{\bOm}[1]{\operatorname*{\Omega}\paren*{#1}}

\newcommand{\bT}[1]{\operatorname*{\Theta}\paren*{#1}}

\DeclarePairedDelimiter\floor{\lfloor}{\rfloor}
\DeclarePairedDelimiter\ceil{\lceil}{\rceil}
\DeclarePairedDelimiter\norm{\lVert}{\rVert}
\DeclarePairedDelimiter\abs{|}{|}
\DeclarePairedDelimiter\brac{\lbrack}{\rbrack}
\DeclarePairedDelimiter\set{\lbrace}{\rbrace}
\DeclarePairedDelimiter\paren{\lparen}{\rparen}

\DeclarePairedDelimiter\ket{|}{\rangle}
\DeclarePairedDelimiterX\braket[2]{\langle}{\rangle}{#1\delimsize\vert#2}

\newtheorem{theorem}{Theorem}
\newtheorem{lemma}[theorem]{Lemma}
\newtheorem{fact}[theorem]{Fact}
\newtheorem{result}[theorem]{Result}

\theoremstyle{definition}

\theoremstyle{remark}

\newcommand{\nth}{^\text{th}}

\newcommand{\yes}{$\mathtt{YES}$}
\newcommand{\no}{$\mathtt{NO}$}

\newtheorem*{BHP}{\bf The $f$-Boolean Hidden Partition problem ($f\operatorname{-BHP}^{\alpha}_{k,n}$)}

{\makeatletter
	\gdef\commentmark{%
		\expandafter\ifx\csname @mpargs\endcsname\relax 
		\expandafter\ifx\csname @captype\endcsname\relax 
		\marginpar{comment}
		\else
		comment 
		\fi
		\else
		comment 
		\fi}
	\gdef\comment{\@ifnextchar[\comment@lab\comment@nolab}
	\long\gdef\comment@lab[#1]#2{{\bf [\commentmark #2 ---{\sc #1}]}}
	\long\gdef\comment@nolab#1{{\bf [\commentmark #1]}}
}

\newcommand{\bhp}[1][f]{\text{$#1$-\normalfont{BHP}}}

\title{Characterizing Quantum Advantage for Generalizations of the Boolean Hidden Matching Problem}
\date{}
\author{
  Mark Bun\footnote{Boston University, \texttt{mbun@bu.edu}} \and
  Joao F. Doriguello\footnote{HUN-REN Alfr\'ed R\'enyi Institute of
  Mathematics, \texttt{doriguello@renyi.hu}} \and
  John Kallaugher\footnote{National University of Singapore,
  \texttt{john.kallaugher@nus.edu.sg}} \and
  Nadezhda Voronova\footnote{IRIF (CNRS \& Universit\'e Paris Cit\'e),
  \texttt{voronova@irif.fr}}
}

\begin{document}
\maketitle

\begin{abstract}
\noindent
We study the one-way communication complexity of the $f$-Boolean Hidden
Partition problem. Here, Alice is given an $n$-bit string and Bob is given
$\Omega(n)$ disjoint blocks of its indices together with a string of labels.
Under the promise that the evaluations of $f$ on these blocks either agree with
all of the labels or disagree with all of them, Bob must determine which is the
case using a single message from Alice. This problem generalizes the Boolean
Hidden Matching and Hidden Hypermatching problems, which capture the special
case where $f$ is the parity function.

We establish classical and quantum communication bounds for $f$-Boolean Hidden
Partition in terms of the sign degree $d$ of $f$, proving a conjecture of
Doriguello and Montanaro (TQC 2020). Their prior work gave logarithmic
communication upper bounds for classical protocols when $d \le 1$ and quantum
protocols when $d \le 2$, as well as polynomial lower bounds for certain
structured functions with larger sign degree. We show that for every $f$ of
sign degree $d \ge 2$, the randomized classical communication complexity of
this problem is $\Theta(n^{1-1/d})$ while its quantum communication complexity
lies between $\Omega(n^{1-2/d})$ and $\bO{n^{1-1/\lceil d/2 \rceil} \log n}$.
This completely characterizes the functions $f$ admitting polynomial quantum
advantage as those for which $d \ge 2$, with a new infinite family of
exponential separations given by the case $d = 2$. 
\end{abstract}

\section{Introduction}

When can a short quantum message be more useful than a much longer classical one? The model of \emph{one-way communication} gives a clean setting in which to study this question. In this model, two parties, Alice and Bob, are given inputs $x$ and $y$ respectively. Alice sends a single message to Bob whose goal is to evaluate a function $F(x, y)$. The efficiency of a classical protocol~\cite{yao1979some} is measured by the number of bits Alice sends, whereas that of a quantum protocol~\cite{yao1993quantum} is the number of qubits that comprise her message. In contrast to the setting of time complexity, we know strong, unconditional results demonstrating the advantage of quantum protocols over classical ones in this and in related communication models~\cite{raz1999exponential,Buhrman2001quantum,bar2004exponential,GKKRW07,verbin2011streaming,montanaro2011new,SWY12,doriguello2019quantum,DoriguelloMontanaro20}.

Bar-Yossef, Jayram, and Kerenidis~\cite{bar2004exponential} introduced the Hidden Matching problem to obtain an exponential separation between quantum and classical one-way communication complexity. In their relational problem, Alice is given an $n$-bit string, Bob is given $\alpha n / 2$ pairs of indices (for some constant $\alpha \in (0, 1)$), and his goal is to output the parity of Alice's bits at one of his given pairs. Gavinsky, Kempe, Kerenidis, Raz, and de Wolf~\cite{GKKRW07} established an exponential separation for a decisional version of this problem (Boolean Hidden Matching) in which Bob also receives $\alpha n / 2$ labels and must distinguish between the case where the parities of pairs of Alice's bits either agree with all of the labels or disagree with all of them. Both problems admit an $\bO{\log n}$-qubit protocol in which Alice transmits the uniform superposition of her bits. By measuring in a basis adapted to his given pairs, Bob can learn the parity of Alice's bits at a random pair. Meanwhile, these problems require message length $\Omega(\sqrt{n})$ to solve classically, matching the protocol in which Alice sends a random sample of her input in the hope of transmitting a matched pair.

Such strong separations seem to hold only for highly structured problems, raising the following questions. What features of the (Boolean) Hidden Matching problems are essential in yielding exponential separations? What modifications enable the construction of broader families of problems for which these separations hold? Replacing pairs of bits by larger blocks of size $k$ yields the Boolean Hidden Hypermatching problem~\cite{verbin2011streaming}. This yields a hierarchy of classical lower bounds, but does not preserve the exponential quantum advantage: to distinguish parities on $k \ge 3$ bits, Shi, Wu, and Yu~\cite{SWY12} proved an $\Omega(n^{1-2/k})$ quantum communication lower bound. Hence, the existence of an exponential separation seems quite sensitive to the predicate being applied to blocks of bits (i.e., the two-bit parity function vs.\ parities on $k$ bits).

Motivated by these questions, Doriguello and Montanaro~\cite{DoriguelloMontanaro20} introduced a generalization of the Boolean Hidden Matching problem called the $f$-Boolean Hidden Partition problem, in which blocks of Alice's bits are combined according to a Boolean function $f$. Suitable choices of $f$ yield an infinite class of problems exhibiting exponential quantum advantage. However, Doriguello and Montanaro could only establish a partial characterization of the functions $f$ for which such a separation holds. In this work, we revisit the $f$-Boolean Hidden Partition problem and close several gaps left by~\cite{DoriguelloMontanaro20}. We give a complete classification of the functions $f$ exhibiting exponential quantum advantage and polynomial quantum advantage. In doing so, we further expand the family of exponential separations between quantum and classical one-way communication complexity.

\subsection{The \texorpdfstring{$f$}{f}-Boolean Hidden Partition Problem}

In this paper, let $[n]:=\{1,\dots,n\}$ for $n\in \mathbb{N} := \{1,2,\dots\}$. Given $x\in\{1,-1\}^n$, let $\overline{x}$ be its bitwise negation, i.e., $\overline{x}_i=-x_i$ for every $i\in[n]$. For a set $S\subseteq[n]$, let $x_S$ be the restriction of $x$ to the coordinates in $S$, i.e., $(x_S)_i=x_i$ for every $i\in S$. 

The $f$-Boolean Hidden Partition problem, denoted by $\bhp_{k,n}^\alpha$, is a one-way communication problem introduced in~\cite{DoriguelloMontanaro20} and is parameterized by a Boolean function $f:\{1,-1\}^k\to \{1,-1\}$, an integer $n\in \mathbb{N}$, and a real number $\alpha\in(0,1)$ such that $\alpha n/k$ is an integer. We will think of the first $\alpha n$ indices in $[n]$ as being divided into $\alpha n/k$ disjoint blocks $(B_i)_{i=1}^{\alpha n/k}$ each of size $k$, where $B_i=((i-1)k+1,\dots,ik)$. For any string $z\in \{1,-1\}^n$ and fixed permutation $\sigma$ on $[n]$, let $z^i$ be the restriction of $z$ to the coordinates in the block $\sigma(B_i)$, i.e., $(z^i)_j=z_{\sigma((i-1)k+j)}$ for every $j\in[k]$.

\begin{BHP}
    Let $\alpha\in(0,1)$ and $n,k\in\mathbb{N}$ such that $k|\alpha n$. Let $f:\{1,-1\}^k\to \{1,-1\}$. Alice is given $x\in \{1,-1\}^n$, and Bob is given a permutation $\sigma$ on $[n]$ and $y\in \{1,-1\}^{\alpha n/k}$. The players are promised that either $f(x^i)=y_i$ for every $i\in[\alpha n/k]$ (a \yes{} instance) or $f(x^i)=-y_i$ for every $i\in[\alpha n/k]$ (a \no{} instance). Only Alice may send a message to Bob. The objective is for Bob to determine which of these cases holds.
\end{BHP}

The Boolean Hidden Matching problem~\cite{GKKRW07} is the special case of $\bhp_{k,n}^\alpha$ where $k = 2$ and $f$ is the parity function. As described above, it exhibits an $\Omega(\sqrt{n})$ vs.\ $\bO{\log{n}}$ separation between classical and quantum one-way communication complexity. The $f$-Boolean Hidden Partition problem generalizes the Boolean Hidden Matching problem by allowing $f$ to be any Boolean function. 

As conjectured and partially substantiated in~\cite{DoriguelloMontanaro20}, the sign degree of $f$ is the key parameter that determines the classical and quantum communication complexity of $\bhp_{k,n}^\alpha$. A polynomial $p:\{1,-1\}^k\to \mathbb{R}$ is said to \emph{sign-represent} $f$ if $p(z)f(z)>0$ for every $z\in\{1,-1\}^k$. The \emph{sign degree} of $f$, denoted by $\mathrm{deg}_\pm(f)$, is the minimum degree of a polynomial that sign-represents $f$. Doriguello and Montanaro~\cite{DoriguelloMontanaro20} showed that if $\mathrm{deg}_\pm(f)\le 1$, then there is a classical protocol for $\bhp_{k,n}^\alpha$ with $\bO{\log n}$ communication, and if $\mathrm{deg}_\pm(f)\le 2$, then there is a quantum protocol for $\bhp_{k,n}^\alpha$ with $\bO{\log n}$ communication. 

On the other hand, \cite{DoriguelloMontanaro20} proved lower bounds in terms of the \emph{pure high degree} of $f$, denoted $\mathrm{phd}(f)$, which is the least degree of a parity function that has nonzero correlation with $f$ under the uniform distribution. Pure high degree is known to be at most the sign degree of $f$~\cite{sherstov2007separating,bun2015dual} but can be much smaller.\footnote{For example, the 3-bit predicate that outputs $1$ iff all input bits agree has sign degree two, but its expectation under the uniform distribution is $-1/2$, so its pure high degree is zero.} Doriguello and Montanaro~\cite{DoriguelloMontanaro20} showed that if $\mathrm{phd}(f) = d \geq 2$, then any classical protocol for $\bhp_{k,n}^\alpha$ requires $\Omega(n^{1-1/d})$ bits, and if $\mathrm{phd}(f) = d \geq 3$, then any quantum protocol for $\bhp_{k,n}^\alpha$ requires $\Omega(n^{1-2/d})$ qubits. However, they left open the question of whether these results are tight and conjectured that their lower bounds could be improved to hold for the sign degree of $f$ instead of the pure high degree. 

\subsection{Our Results}

Our results confirm sign degree as the key parameter governing the quantum and classical one-way communication complexities of all problems in the $\bhp$ family. Sign degree determines classical communication complexity up to constant factors, characterizes the functions $f$ admitting exponential quantum advantage, and gives quantum bounds that are tight up to logarithmic factors whenever the sign degree is even.

 In one direction, we resolve the conjecture of~\cite{DoriguelloMontanaro20}  by proving classical and quantum communication complexity lower bounds for $\bhp_{k,n}^\alpha$ in terms of the sign degree of $f$.
\begin{result}
    For every $f:\{1,-1\}^k\to \{1,-1\}$ with sign degree $d\ge 2$, the randomized classical communication complexity of $\bhp_{k,n}^\alpha$ is $\Omega(n^{1-1/d})$. Given $f:\{1,-1\}^k\to \{1,-1\}$ with sign degree $d\ge 3$, the quantum communication complexity of $\bhp_{k,n}^\alpha$ is $\Omega(n^{1-2/d})$.
\end{result}

In the other direction, we prove upper bounds for $\bhp_{k,n}^\alpha$ in terms of the sign degree of $f$ beyond the case $\mathrm{deg}_\pm(f) \leq 2$ from~\cite{DoriguelloMontanaro20}.
\begin{result}
    For every $f:\{1,-1\}^k\to \{1,-1\}$ with sign degree $d\ge 2$, the randomized classical communication complexity of $\bhp_{k,n}^\alpha$ is $\mathrm{O}(n^{1-1/d})$, while its quantum communication complexity is $\mathrm{O}(n^{1-1/\lceil d/2\rceil}\log n)$.
\end{result}

Putting these results together, we completely determine the classical complexity of $\bhp_{k,n}^\alpha$ as $\Theta(n^{1-1/d})$ whenever $d=\mathrm{deg}_\pm(f) \ge 2$. Meanwhile, its quantum complexity is $\widetilde{\Theta}(n^{1-2/d})$ for even $d$ and lies between $\Omega(n^{1-2/d})$ and $\mathrm{O}(n^{1-1/\lceil d/2\rceil}\log n)$ for odd $d$.

In combination with the prior results of~\cite{DoriguelloMontanaro20}, our bounds divide the family of $\bhp$ problems into three main regimes depending on the sign degree of $f$. When $\mathrm{deg}_\pm(f) \le 1$, logarithmic classical communication already suffices. When $\mathrm{deg}_\pm(f) = 2$, the classical complexity is $\Theta(\sqrt{n})$ while there is a  quantum protocol using only $\bO{\log n}$ qubits. When $\mathrm{deg}_\pm(f) \ge 3$, quantum communication itself must be polynomial in $n$, but always remains polynomially smaller than classical communication. Thus, the functions $f$ of sign degree two are exactly those yielding an exponential separation between classical and quantum communication complexity. Moreover, every sign degree at least two yields (at least) a polynomial quantum advantage. 

\paragraph{Technical Contribution.} We prove our lower bounds using Boolean Fourier analysis and hypercontractivity, which are standard tools for one-way communication lower bounds going back to~\cite{GKKRW07}.
However, to the best of our knowledge, they have never\footnote{There are two
similar-sounding, but technically very distinct scenarios to distinguish from
ours here. The first is where the hardness of the communication problem is
proved for a product distribution, but then a reduction subsamples Bob's input
in a way that depends on Alice's input (e.g., \cite{KKS15,CGSV24}). A second is
the setting of~\cite{KKP18} in which multiple players have uniform inputs that
obey some parity constraints jointly, and all send messages to a referee. This
is morally more similar but the actual effect is to have the referee consider a
different but still independent set of parities.} been used to prove lower
bounds in the case where the hard distribution requires Alice's input to depend
on Bob's choice of parities, and the standard set of tools breaks as soon as the
assumption of independence is lost. Our technical contribution in this paper is
a set of tools that allow significant correlation between Alice and Bob. We
expect these to be of use in analysing other one-way problems that are too easy
for independent inputs.

\subsection{Proof Overview: Lower Bound}
To understand our lower bound, it helps to understand why we would expect it to
hold in the first place. Suppose Alice sends $|S|$ of her bits to Bob at random.
For every parameter $r\leq |S|$, this yields $\binom{|S|}{r}$ $r$-tuples of bits.
In expectation, $\frac{\alpha n}{k} \cdot \binom{|S|}{r} \cdot \binom{k}{r}
\cdot \binom{n}{r}^{-1} = \bT{\frac{|S|^r}{n^{r-1}}}$ of these tuples land entirely
within some block. So at least for this class\footnote{It is helpful for
intuition to think of the corresponding quantum class of protocols as sampling
pairs at the same cost at which the classical protocols sample singletons. A
simple quantum protocol therefore allows Bob to learn the parity of $d$ bits from within at
least one block at the cost of $\bO{n^{1 - 1/\ceil{d/2}}\log n}$ qubits sent from
Alice.} of protocols, $\Theta(n^{1 - 1/d})$ bits is the threshold for successfully
teaching Bob a size-$d$ parity of any block. But why is this necessary? After all, a function with high sign degree can still be determined by a
single bit, e.g., the $d$-bit function given by $x_1 \wedge (x_2 \oplus
\dots \oplus x_d)$. 

The reason comes from the dual (in the sense of linear programming duality) characterization of sign degree. If $f$ has sign
degree $d$, then there is a distribution $q$ on $\bools^k$ such that
\[
	\E[z\sim q]{f(z)\chi_S(z)}=0
	\qquad\text{for every }|S|<d.
\]
Therefore, if Alice's input on a given block $i$ were sampled from $q$, any subset of its bits of size $< d$ would not provide any meaningful information to Bob since his
posterior on $f(x^i)$ would be uniform even conditioned on those bits.

We now need a way to generalize this argument to protocols where Alice sends
arbitrary classical or quantum messages. In both cases, the standard approach
runs through Boolean Fourier analysis via (in the quantum case, matrix)
hypercontractivity~\cite{bonami1970etude,KKL88,BARdW08}. These tools (which we describe in more detail later) allow us to show that Bob learns not much more about parities than he would have if Alice
had just sent random bits. Loosely speaking, we can say that if Alice sends Bob a
$c$-bit message, he learns at most $\binom{\bO{c}}{r}$ ``about'' size-$r$
parities\footnote{Formally this is quantified in terms
of the squares of his posterior expectations on those parities; in the sampling
case he knows a parity with certainty or not at all so these are $0$-$1$
values.}. We can then exploit
the randomness of which parities Bob cares about (determined by
the permutation $\sigma$ in our case) and argue that if Alice sends too short of a message, Bob
will not learn enough about size-$d$ parities to change his posterior on
any one block or to change it a little on many blocks. In the quantum case, the
tools are somewhat different, but the structure of the argument is
similar.

However, a problem immediately arises. The argument above implicitly assumes
that the parities Bob cares about are independent of Alice's input and,
consequently, the message she sends. For the special cases of Boolean Hidden (Hyper)Matching, this is true: the ``hard distribution'' is for Alice to receive
a uniformly random string, uncorrelated with the parities Bob will ask about.
But if we want Alice's input on a block to be drawn from $q$, this can no
longer be the case. In fact, the block permutation $\sigma$ can be very
correlated with Alice's input. There could be, in general, message-permutation
pairs that are $2^{\bT{n}}$ times more likely than the product of their
marginal probabilities, and many more that have zero probability of arising
together.

This motivates the way we choose the rest of our hard distribution $Q$ (see
\cref{sec:harddist}). The permutation $\sigma$ on $\brac{n}$ is chosen
uniformly at random\footnote{Note that it is still not revealed to Alice, which
would make the problem trivial.}, the blocks $(x^i)_{i \in \brac{\alpha n/k}}$
are sampled independently from $q$, and every other bit is sampled uniformly at
random. 

If we want to bound how much Bob learns about $f(x^i)$ in expectation, we can
condition on the restriction of $\sigma$ to every other block (but not the
uniform bits), and then the correlation between Alice's input on block $i$ and
the uniform bits with the remaining permutation is bounded by a constant
depending only on $q$. However, this is not quite enough, for Bob could learn
nothing about any $f(x^i)$ and still solve the problem. For instance, if Alice
successfully sends him the value of $f(x^1)f(x^2)f(x^3)$, he could compare this
with $y_1y_2y_3$. Traditionally, this has been solved by analyzing how much
information Alice sends about higher-order parities of her input; the same
argument that Bob only learns $\bO{\frac{c^d}{n^{d-1}}}$ about individual
blocks generalizes to showing that he only learns $\bO{\frac{c^{2d}}{n^{2d -
2}}}$ about pairs of blocks, and so on. But this option is not available to us:
our method for bounding the correlation between $\sigma$ and Alice's input
breaks down when we start considering many blocks at a time. If we want to
bound how much Bob learns about $\prod_{i \in U}f(x^i)$, we can condition on
$\sigma$ restricted to the blocks $B_i$ for $i \not\in U$, but the remainder of
$\sigma$ will still have $2^{\bT{\abs{U}}}$ correlation with Alice's input.

Instead, we imagine revealing $f(x^1),f(x^2),\ldots$ one block at a time and ask how much Bob learns about $f(x^i)$ conditional on
$(f(x^j))_{j<i}$ (but before he looks at $y_i$). If Bob has learned anything
about $(f(x^i))_{i=1}^{\alpha n/k}$, he must have learned something about
these conditional single-block values, even if he has not learned anything
about any given block unconditionally. 

Thus both proofs reduce to the question:
\textit{Suppose Bob already knows the values of $f$ on the previous blocks. How
much can Alice's short message tell him about $f$ on the next random block?}

For any given message, this can be dramatically more than what Bob learns
unconditionally---consider the case where Alice has successfully informed him about the XOR
of three blocks. But on average across messages, it cannot be. Suppose that
before the game starts, Alice and Bob already knew that $f(x^1) = \dots =
f(x^{\alpha n/k - 1}) = -1$. Suppose that beyond that, they knew the value of
$\sigma$ restricted to the first $\alpha n / k - 1$ blocks. Then $f(x^{\alpha n
/k})$ would still be uniformly distributed, and the challenge of designing a
protocol would not be significantly reduced: Alice would still have
$(1-\alpha)n + k$ bits of which an unknown $k$ were important to Bob.

The quantum and classical proofs analyse this conditional single-block
expectation using closely related Fourier-analytic ideas. The quantum proof is
relatively direct. The classical proof has one additional obstacle caused by
the fact that $q$ may not have full support.

\subsubsection{Quantum Lower Bound}
Write $m(x)$ for the $c$-qubit density matrix corresponding to Alice's message
when her input is $x$. For any given $(\sigma, y)$, Bob's best possible
advantage over guessing at random is determined by the trace distance (up to an appropriate normalization) \[
\norm*{\sum_{\substack{x \in \bools^n:\\\text{$x$ from a \yes{} instance}}}Q(x)m(x) -
\sum_{\substack{x \in \bools^n:\\\text{$x$ from a \no{} instance}}}Q(x)m(x)}_1.
\]
Through appropriate re-arrangement, we can bound this by the sum, over all $i
\in \brac{\alpha n/k}$, of \[
\norm*{\E[x \sim Q| \sigma, \Fc_z]{f(x^i)m(x)}}_1 
\]
where $\Fc_z$ denotes an event conditioning on $f(x^j) = z_j$ for every $j < i$.
We show (\cref{lm:qfbounds}) that regardless of the value of $\Fc_z$ conditioned on, the expectation of this quantity over $\sigma$ is bounded by
\[
	\bO{\paren*{\frac{c}{n}}^{d/2} }.
\]
We need two ideas to make this work.

\paragraph*{Conditioning away other blocks.} We convert the expectation above
into a $((1 - \alpha)n + k)$-bit function by fixing the restriction of
$\sigma$ to blocks other than $i$ and averaging over all possibilities for
Alice's inputs to those blocks that result in $\Fc_z$ holding. Writing $S$ for
the remaining $(1-\alpha)n + k$ indices, this yields a new matrix-valued
function \[
\nu:\bools^S\to \mathcal{D}(\mathbb{C}^{2^c}), 
\]
which still takes values on $c$-qubit density matrices. Conditioning away other blocks
formalizes the intuition that a version of the game in which Alice and Bob know
the values of blocks $1$ through $i-1$ a priori would be equally hard,
\emph{while} maintaining $\nu$ as a density matrix, which allows using the matrix hypercontractivity
bound~\cite{BARdW08} \[
\sum_{U:\abs{U} = r}\norm{\wh{\nu}(U)}_1^2 \le 4c^r.
\]
Here the Fourier transform of $g:\{1,-1\}^k \to \mathcal{D}(\mathbb{C}^{2^c})$ is defined as $\widehat{g}(U) = \frac{1}{2^k}\sum_{x\in\{1,-1\}^k} g(x)\chi_U(x)$, where $\chi_U(x) = \prod_{i\in U} x_i$ is the parity function on the indices in $U$. In other words,
this averaged-and-conditioned density matrix is still bounded in the amount it
can tell Bob about individual parities of Alice's input (outside of the
irrelevant non-$i$ blocks).

Conditioned on everything we fixed, the $k$ coordinates of the current block
are still placed uniformly at random among these remaining coordinates, and so
in particular any
given parity $U$ \emph{from block $i$} is equally likely to land anywhere, meaning that \[
\norm*{\wh{\nu}(U)}_1^2 \le \paren*{\frac{\bO{c}}{n}}^{\abs{U}}
\] 
in expectation over the not-already-fixed parts of $\sigma$. It remains to use this insight to bound \[
\norm*{\E[x \sim Q| \sigma, \Fc_z]{f(x^i)m(x)}}_1.
\]

\paragraph*{Associating the conditional expectation of $fq$ with degree $\ge d$
parities.} Because we averaged away all but one block where $Q$ was not the
identity, the above expression becomes \[
\norm*{\E[x \sim Q| \sigma, \Fc_z]{f(x^i)m(x)}}_1 =  \norm*{\E[u \sim \Uc]{2^kq(u^i)f(u^i)\nu(u)}} \le \bO{1}\sum_{U
\subseteq \brac{k} : \abs{U} \ge d}\norm{\wh{\nu}(\sigma(U))}_1
\]
by expanding $fq$ as a sum of Fourier coefficients, which are $0$ below degree $d$.

As our bound is on \emph{squares} of Fourier coefficients, we incur a Cauchy-Schwarz loss at this point, bounding \[
\E[\sigma]{\norm{\wh{\nu}(\sigma(U))}_1} \le
\E[\sigma]{\norm{\wh{\nu}(\sigma(U))}_1^2}^{1/2} \le
\paren*{\frac{\bO{c}}{n}}^{\abs{U}/2}
\]
and so concluding that \[
\norm*{\E[x \sim Q| \sigma, \Fc_z]{f(x^i)m(x)}}_1 \le \bO{\paren*{\frac{c}{n}}^{d/2}}
\]
in expectation over $\sigma$ as required, giving us a $c = \Omega(n^{1 - 2/d})$
threshold for this to be $\Omega(1/n)$ and therefore for Bob's total advantage
to be constant.

\subsubsection{Classical Lower Bound}
The classical lower bound starts with a similar reduction to conditional
single-block expectations. Let $m : \bools^n \rightarrow \bool$ denote the
expression of Alice's $c$-bit message as an indicator function\footnote{As
Alice and Bob know the input distribution, we may without loss of generality
assume they use a deterministic protocol, as any randomized protocol converts
to an at-least-as-good deterministic protocol by fixing the random seed at
whichever one maximizes success probability over the input distribution.} (so
$m(x) = 1$
iff Alice having the input $x$ would result in her sending the message $m$). We
use Pinsker's inequality and the chain rule for
relative entropy to bound Bob's advantage over random guessing as the square
root of the sum over all $i \in \brac{\alpha n / k}$ of \[
  \E[\sigma,m\mid \Fc_z]{  \E[x\sim p\mid \Fc_z, m(x) = 1]{f(x^i)}^2 },
\]
where $p$ is Bob's posterior distribution on Alice's inpute before he looks at $y$, and as in the quantum case, $\Fc_z$ denotes conditioning on $f(x^j) = z_j$
for every $j < i$. We show in \cref{lm:fbounds} that this is bounded by \[
\bO{\paren*{\frac{c}{n}}^d}
\]
when Alice sends $c < n$ classical bits, giving a $c = \Omega(n^{1 - 1/d})$ lower
bound for Bob to succeed with $2/3$ probability.

Note the distinction from the quantum case: where the quantum lower bound
argument dealt with the linear quantity
$\E[\sigma]{\norm{\wh{\nu}(\sigma(U))}_1}$ and accepted a Cauchy-Schwarz loss,
here we are bounding a quadratic quantity. This is necessary to obtain a tight
lower bound, but it introduces a major technical challenge, as we will be
bounding it in terms of quantities that can potentially saturate Markov's
inequality.

\paragraph{Conditioning away other blocks.}
Our first step is very similar to the quantum one. Condition on $\sigma$
restricted to the non-$i$ blocks, and define a $((1-\alpha)n + k)$-bit function
$\nu$ by averaging $m(x)$ over all possibilities for Alice's inputs to these
blocks that result in $\Fc_z$ holding. Writing $S$ for
the remaining $(1-\alpha)n + k$ indices, this gives us a function \[
\nu : \bools^S \rightarrow \brac{0,1}
\]
that defines Bob's posterior on $f(x^i)$ by \[
\E[x\sim Q \mid \Fc_z, m(x) = 1]{f(x^i)} = \frac{\sum_{x \in
\bools^S}f(x^i)\nu(x)q(x^i)}{\sum_{x \in \bools^S}\nu(x)q(x^i)}
\]
as block $i$ is the only part of $S$ where $Q$ is not the uniform distribution.

Similarly to the quantum case, expanding $qf$ in the Fourier basis gives \[
\E[x\sim Q \mid \Fc_z, m(x) = 1]{f(x^i)}^2 \le \bO{\Nc^{-2}}\sum_{U
\subseteq \brac{k} : \abs{U} \ge d} \wh{\nu}(\sigma(U))^2,
\]
where $\Nc$ denotes the expectation of $q(x^i)\nu(x)$. At first glance, this is
very promising: use the classical hypercontractivity theorem mentioned earlier
to bound $\sum_{U:\abs{U} = d} \wh{\nu}(\sigma(U))^2$ and
thereby bound Bob's squared posterior in expectation over $\sigma$ restricted
to $S$.

However, there is an issue. For all constant $r$, the hypercontractivity bound
characterizes the ratio of the squared coefficients of $\wh{\nu}$ to the
expectation of $\nu$ itself, for which we write $\Qc$, i.e., \[
\Qc^{-2}\sum_{U:\abs{U} = r} \wh{\nu}(U)^2 \le \bO{\paren*{\frac{C}{r}}^r},
\]
where $C = -\log \Qc$. Similar to the quantum case, here the Fourier transform of $g:\{1,-1\}^k \to \mathbb{R}$ is defined as $\widehat{g}(U) = \frac{1}{2^k}\sum_{x\in\{1,-1\}^k} g(x)\chi_U(x)$. Moreover, $C$ will be $\bO{c}$ with very good probability and will
not cause us major problems, so we will think of it as $c$ from hereon in.

The above means that the Fourier coefficient bound we use has a factor of
$(\Qc/\Nc)^2$, where this ratio corresponds to how much multiplication by
$q(x^i)$ decreases the expected value of $\nu(x)$. If $q$ had full support this
would be a constant, but many natural $f$ (including the example we mentioned
earlier in which a single bit fixes $f$) require $q$ with partial support, and
so the ratio can be arbitrarily large. We have to manage the probability that it is.

We pause here to ask why this was not an issue for the quantum lower bound.
While we do not use this notation explicitly, the ratio $\Qc/\Nc$ does actually
appear implicitly in the density matrix calculation we used. However, as we
were using linear terms, this is mitigated by the fact that $\Nc$ corresponds
to the (conditional) probability that Alice sends the message $m$ that would
generate the averaged function $\nu$. So the effect is to replace weighting by
$\Nc$ with weighting by $\Qc$, and while $\Qc/\Nc$ may be arbitrarily large for
any given message, the sum over all $\Qc$ is exactly $2^k$ times the sum over
all $\Nc$, as $\sum_{x \in \bools^k} 1 = 2^k \sum_{x \in \bools^k}q(x)$.

So to fix the hole in our proof, we need to find a way to replace this
expression with one that has only a linear term for $\Qc/\Nc$. We do so with
\cref{lm:recbound}.

\paragraph{Bounding \texorpdfstring{$\Qc/\Nc$}{N/Q}.}
The first observation is that this ratio itself can be bounded in terms of
Fourier coefficients.  Expanding $q$ as a sum of Fourier coefficients gives
\[
	\Nc = 2^{-k}\Qc + \sum_{U \subseteq \brac{k} : \abs{U} > 0} \wh{q}(U)\wh{\nu}(\sigma(U))
\]
and so the ratio is bounded by $2^{k+1}$ unless some non-empty $U$ has
$\wh{\nu}(\sigma(U))^2 = \bOm{\Qc^2}$ (for a constant dependent on $k$). This has the correct normalization
factor and so the probability that it happens for $\abs{U} = r$ can be bounded by \[
\bO{\paren*{\frac{c}{n}}^r}
\]
over $\sigma$. However, this is still not good enough. As a set with $\abs{U} =
1$ could make the condition hold, we could end up with a $\bO{c/n}$ error term,
which would be far too large for our purposes. We deal with this problem by
recursion, showing that conditional on $\wh{\nu}(\sigma(U))^2 = \bOm{\Qc^2}$
for any given $\abs{U} = r < d$, the expectation of \[
\Nc\E[x\sim Q \mid \Fc_z, m(x) = 1]{f(x^i)}^2
\]
over $\sigma$ is bounded by $\bO{\Qc \paren{\frac{c}{n}}^{d - r}}$.

\paragraph{Recursion.}
Suppose $\wh{\nu}(\sigma(U))^2 = \Qc^2$ for $\abs{U} = r < d$. That would
correspond to $\nu$ fixing the parity $\chi_{\sigma(U)}$: every $x$ such that $\nu(x) >
0$ has $\chi_{\sigma(U)}(x)$ take the same value. For an even stronger supposition,
suppose $\nu$ fixes $x_{\sigma(U)}$ entirely as some $t \in \bools^U$. Then we
could write down restricted functions $q_t$, $f_t$, $\nu_t$ based on fixing the
$U$ bits of the input to be $t$, and rewrite the expectation in terms of those
functions.  In this case we would have $\E[x]{\nu(x)} = 2^{-r} \E[x]{\nu_t(x)}$
and, moreover, for every $V$ with $\abs{V} \le d - r$, \[
\E[x]{q_t(x)f_t(x)\chi_V(x)} = 0.
\]
Then, by adding $U$ to the indices we fix $\sigma$ on and taking expectations
over the remaining permutation $\sigma'$, we could recurse on $q_t, f_t, \nu_t$
to conclude that \[
\E[x\sim Q \mid \Fc_z, m(x) = 1, \sigma|_U]{f(x^i)}^2 \le
\frac{\Qc_t}{\Nc_t}\bO{\left(\frac{c}{n}\right)^{d-r}} = \frac{\Qc}{\Nc}\bO{\left(\frac{c}{n}\right)^{d-r}}
\]
in expectation over $\sigma'$, where the remaining $\frac{\Qc}{\Nc}$ term
measures the correlation between $\sigma$ and Alice's message, as once the
blocks other than $i$ have been conditioned on, $\sigma$ controls the location
of block $i$ and makes a given message $m$ more likely by a factor of \[
2^k\frac{\E[x]{q(x^i)\nu(x)}}{\E[x]{\nu(x)}} = 2^k\frac{\Nc}{\Qc},
\]
where $\Nc$ can be seen as measuring the likelihood of Alice drawing an $x$
resulting in message $m$ after block $i$ is known, and $\Qc$ as that likelihood
before block $i$ is known (recall that $\Qc$ and $\Nc$ both depend on the
message $m$).

Now, in general, $\wh{\nu}(\sigma(U))^2 = \bOm{\Qc^2}$ does not imply $\nu$
fixes $x_{\sigma(U)}$. However, we can still recover the recursive step above
by writing \[
\E[x\sim Q \mid \Fc_z, m(x) = 1, \sigma|_U]{f(x^i)}^2
\]
as a weighted sum over cases where $x_{\sigma(U)}$ is fixed to various values
of $t \in \bools^U$. This suffices to recover the recursion and prove
\cref{lm:recbound}.

\paragraph{Averaging over messages.} Even with \cref{lm:recbound}, some
messages will generate a very large posterior expectation. For instance, there
could be a message that Alice sends on exactly one input, and so given that
message, Bob would have complete certainty about $f(x^i)$, regardless of
$\sigma$. But as we defined $\nu$ by averaging over Alice's inputs for blocks other
than $i$ (using only those that made $\Fc_z$ hold), the probability that Alice sends any given message conditioned on
$\sigma$ and $\Fc_z$ is exactly \[
\sum_{x \in \bools^S} \nu(x) q(x^i) 2^{k-|S|} = 2^k \Nc
\]
and so averaging over all messages gives us \[
\E[m,\sigma]{\E[x\sim Q \mid \Fc_z, m(x) = 1, \sigma|_U]{f(x^i)}^2} \le
\E[\sigma]{\sum_m \frac{\Qc}{\Nc} 2^k\Nc \bO{\paren*{\frac{c}{n}}^d}} \le
\bO{\paren*{\frac{c}{n}}^d}
\]
because each $\nu$ came from averaging a different message, and the messages
partition $\bools^n$. So if $c = \lO{n^{1 - 1/d}}$, Bob's squared
conditional posterior expectation on each $f(x^i)$ is $\lO{1/n}$ and so his
advantage over random guessing is $\lO{1}$.

\subsection{Proof Overview: Upper Bound}
Our upper bound comes from the sign-representing polynomial characterisation of
$f$. There is a degree-$d$ polynomial $p(x)$ such that $p(x)$ always has the
same sign as $f(x)$, and so as $k$ is a constant there is some threshold $\beta$
such that $p(x) > \beta$ whenever $f(x) = 1$ and $p(x) < -\beta$ whenever $f(x)
= -1$. This implies that that in a \yes{} instance of the game, \[
\sum_{i=1}^{\alpha n /k } y_ip(x^i) > \beta \alpha n/k
\]
while in a \no{} instance, the opposite relation holds. As $p$ is a degree-$d$
polynomial, this sum can be expanded as a sum of parities up to degree $d$,
each contained within a single block, and with a constant range of minimum and
maximum coefficient. So, we can rewrite this sum as \[
\sum_{U \subseteq \brac{k}} \alpha_U \sum_{i = 1}^{\alpha n/k}y_i\chi_U(x^i)
\]
for a collection of constants\footnote{For the purpose of this discussion we
are not interested in the dependence on $f$ or $k$.} $\alpha_U$. Therefore, the
difficulty of estimating this sum reduces to the difficulty of sampling
uniformly from the parities $\chi_U$ for any given $U \subseteq \brac{k}$, and
if we can sample $\bT{1}$ of these copies for each such $U$ we can obtain a
sufficiently accurate estimate. For the rest of this section we describe the
case where $\abs{U} = d$; the problem becomes easier as $\abs{U}$ shrinks. 

In the classical case, we have the bound immediately, up to a log
factor\footnote{This log factor can be removed if Alice instead uses a
constant-wise independent hash function to choose which bits to send and in
which order, and then sends Bob the identity of the hash function, rather than
sending him the actual indices.}: Alice sends Bob the identity of $c$ of her
bits, and so he learns any given size-$d$ parity with probability $(c/n)^d$.
She thus need only send him $\Theta(n^{1-1/d})$ bits to teach him the parity
of $\bT{1}$ copies of $U$, as $\alpha$ is constant.

In the quantum case, Alice has access to a more efficient way of communicating
size-2 parities to Bob\footnote{More precisely, to communicate parities among
the unknown-to-her list of parities that Bob is looking for.}. As in the
(Boolean) Hidden Matching protocol of \cite{bar2004exponential,GKKRW07}, she may send the quantum state
\[\frac{1}{\sqrt{n}}\sum_{i = 1}^n x_i \ket{i}\]
and then, taking any list of size 2 parities $(a_j,b_j)$ such that one is in
each block, Bob measures in a basis that contains \[
\frac{\ket{a_j} + \ket{b_j}}{\sqrt{2}}, \frac{\ket{a_j} - \ket{b_j}}{\sqrt{2}}
\]
for every $j$. With probability $\bOm{1}$, one of the $2\alpha n/k$ basis
elements listed above will be measured, and reveal the corresponding parity to
Bob as each of the two states is compatible with exactly one of the two
possible parities $x_{a_j}x_{b_j}$.

Alice and Bob divide $U$ into $\floor{d/2}$ pairs, and possibly a singleton.
Then for each pair $P$ they run the size-two parity sampling procedure above to
recover the parity of $P$ for a random $P$ among their blocks with constant
probability. If they have a singleton they can likewise have Alice sample a
single bit at random and send it. This procedure costs $\bO{\log n}$ qubits, and
for Bob to successfully learn a full copy of $U$, one needs
$\floor{d/2}$ of its pairs sampled, along with the possible singleton. If the
procedure is repeated $c$ times, this happens with probability \[
\bT{c/n}^{\floor{d/2}} \cdot \bT{c/n}^{d - 2\floor{d/2}} = \bT{(c/n)^{\ceil{d/2}}}.
\]
They will therefore succeed with constant probability after Alice sends
$\mathrm{O}(n^{1 - 1/\ceil{d/2}}\log n)$ qubits to obtain a $\bT{1/n}$
per-block chance. Repeating the entire process gives any constant number of
samples, allowing Bob to estimate the desired quantity.

\section{The \texorpdfstring{$f$}{f}-Boolean Hidden Partition Problem}
\label{sec:probdef}
Our parameters are $n, k \in \Nbb$ and $\alpha \in (0,1)$, such
that $\alpha n / k$ is an integer. Let $f :
\bools^k \rightarrow \set{1, -1}$ be a Boolean function.  An instance of the
game consists of a string $x \in \bools^n$ given to Alice, and a permutation
$\sigma$ on $\brac{n}$ and string $y \in \bools^{\alpha n / k}$, both given to Bob.

We will think of $\brac{\alpha n}$ as being divided into ``blocks''
$(B_i)_{i=1}^{\alpha n /k}$, where $B_i = \set{(i-1)k + 1, \dots, ik}$. The
permutations of these blocks under $\sigma$ will induce substrings of an $n$-bit
string: for any $z \in \bools^n$, $z^i$ with $i\in[\alpha n/k]$ denotes the $k$-bit string whose $j\nth$ bit
is $z_{\sigma((i-1)k + j)}$.

\paragraph{Problem.} The players are promised that either $f(x^i) = y_i$ $\forall i \in
\brac{\alpha n/k}$ (a \yes{} instance) or $f(x^i) = -y_i$ $\forall i \in
\brac{\alpha n/k}$ (a \no{} instance). Alice may send a message
to Bob, but Bob may not reply. The objective of the game is for Bob to guess
which of these cases holds.

\paragraph{Additional notation.}
For any set $U$, positive integer $m$, $x \in \bools^m$, and $\pi : U \rightarrow \brac{m}$, we write $x^\pi_U$ for the string in $\bools^U$ given by \[
(x^\pi_U)_i = x_{\pi(i)}.
\]

When we have two disjoint sets $R_1, R_2 \subseteq \brac{n}$ and strings $x \in
\bools^{R_1}, y \in \bools^{R_2}$, we will generalise the concatenation
operation $x \cdot y$ to mean $(x \cdot y) \in \bools^{R_1 \cup R_2}$ with \[
(x \cdot y)_i = \begin{cases}
x_i & i \in R_1\\
y_i & i \in R_2.
\end{cases}
\]

\paragraph{Norms.}
Throughout, for any function $h : \bools^m\rightarrow \Rbb$ and $p \in
\Rbb_{\ge 1}$, \[
\norm{h}_p = \E[x \sim \Uc(\bools^m)]{\abs{h(x)}^p}^{1/p} \qquad\text{and}\qquad \hat{\|}h\hat{\|}_p = \left(\sum_{S \subseteq [m]} \abs{\widehat{h}(S)}^p\right)^{1/p}.
\]
For a matrix $A$, $\norm{A}_p$ will be the Schatten $p$-norm.

\paragraph{Sign degree.}
The \emph{sign degree} of a Boolean function $f : \bools^n \to \bools$, denoted as $\mathrm{deg}_\pm(f)$, is the least total degree of a real polynomial $p : \bools^n \to \Rbb$ such that
\[p(x)f(x) > 0 \quad \text{for every }x\in \bools^n.\]
Any $p$ that satisfies the condition above is said to \emph{sign-represent} $f$.
Sign degree $\mathrm{deg}_\pm(f)$ is sometimes called the \emph{threshold degree} of $f$.

\begin{fact}[Dual characterization of sign degree]
Let $f : \bools^n \to \bools$. Then $\mathrm{deg}_\pm(f) \geq d$ if and only if there exists a probability distribution $q: \bools^n \to [0, 1]$ such that
\[\E[x \sim q]{f(x) \chi_S(x)} = 0 \quad \text{for every } S \subseteq [n], |S| < d.\]
\end{fact}
Here $\chi_S$ is the parity function $\chi_S(x) = \prod_{i \in S}x_i$.

\section{Lower Bounds}
\subsection{Hard Distribution}
\label{sec:harddist}
Let $q : \bools^k \rightarrow \Rbb$ be any distribution under which $f$ is
uncorrelated with all its degree $<d$ parities. If $f$ is a function of
sign-degree $\ge d$ this is guaranteed to exist by the dual definition of
sign-degree. We define \[
Q(x) = 2^{(\alpha-1)n}\prod_{i=1}^{\alpha n/k}q(x^i).
\]
Note that $Q$ depends on $\sigma$. Our hard distribution will be:
\begin{itemize}
\item $\sigma$ is sampled uniformly at random;
\item $x$ is sampled from $Q$;
\item $y$ is chosen so that either $f(x^i) =
y_i$ $\forall  i \in \brac{\alpha n/k}$ or $f(x^i) = -y_i$ $\forall i \in \brac{\alpha n/k}$, with equal
probability.
\end{itemize}

\subsection{Fourier Analysis and Hypercontractivity}

Any Boolean function $f:\{1,-1\}^k\to\mathbb{C}$ can be decomposed as
\begin{align*}
    f(x) = \sum_{U\subseteq[k]} \widehat{f}(U)\chi_U(x) \quad\text{where}\quad \chi_U(x) = \prod_{i\in U} x_i.
\end{align*}
The coefficients $\widehat{f}(U) = \frac{1}{2^k} \sum_{x\in\{1,-1\}^k} f(x) \chi_U(x)$ are called Fourier coefficients of $f$. The same decomposition is true if $f:\{1,-1\}^k \to \mathbb{C}^{k\times k}$ is defined as a map to $k\times k$ complex matrices, the difference being that the Fourier coefficients $\widehat{f}(U)$ are now also complex matrices, so that the usual Fourier decomposition (onto $\mathbb{C}$) is performed entry-wise.

We will need a slight variant of the standard KKL corollary~\cite{KKL88} used in one-way
communication complexity. We start by stating the hypercontractive inequality. For any $\delta \in \brac{0,1}$, let $T_\delta$ be the ``noise operator'' that
takes $f(x)$ to the expectation over $f(y)$ when $y$ is chosen by flipping each
bit of $x$ independently with probability $1/2 - \delta/2$. We will need the
fact that $\wh{T_\delta f}(U) = \delta^{\abs{U}} \wh{f}(U)$.
\begin{fact}[\cite{bonami1970etude}]\label{fact:bonami}
Let $f : \bools^n \rightarrow \Rbb$ and let $1 \le p \le q \le \infty$. Then
$\norm{T_\delta f}_q \le \norm{f}_p$ for $0 \le \delta \le \sqrt{\frac{p-1}{q-1}}$.
\end{fact}
\begin{lemma}
\label{lm:genkkl}
For any $m \in \Nbb_{>0}$, let $\nu : \bools^m \rightarrow \brac{0,1}$ be a
function with $\E[x]{\nu(x)} = \Qc > 0$ and $C = -\log \Qc$. Then \[ \Qc^{-2}
\sum_{U \subseteq \brac{m} : \abs{U} = r} \wh{\nu}(U)^2 \le
\paren*{\frac{4C}{r}}^r \quad\text{for all integers $r$ such that } 1 \le r \le C.
\]
\end{lemma}
\begin{proof}
Set $\delta = r/C$. Then
\begin{align*}
\Qc^{-2} \sum_{U \subseteq \brac{m}: \abs{U} = r}  \wh{\nu}(U)^2
&= \frac{1}{\delta^r}\Qc^{-2} \sum_{U \subseteq \brac{m}: \abs{U} =
r} \wh{T_{\sqrt{\delta}}\nu}(U)^2\\
&= \frac{1}{\delta^{r}} 2^{2C}\sum_{U \subseteq \brac{m} :
\abs{U} = r} \wh{T_{\sqrt{\delta}} \nu}(U)^2\\
&\le \frac{1}{\delta^{r}} 2^{2C}\norm{T_{\sqrt{\delta}} \nu}_2^2\\
&\le \frac{1}{\delta^{r}} 2^{2C}\norm{\nu}_{1 + \delta}^2 \tag{by \Cref{fact:bonami}}\\
&\le \frac{1}{\delta^{r}} 2^{2C}\norm{\nu}_1^{\frac{2}{1 + \delta}} \tag{$\nu(x)\in[0,1]$}\\
&= \frac{1}{\delta^{r}} 2^{2C(1 - \frac{1}{1 + \delta})}\\
&\le \paren*{\frac{4C}{r}}^r.\qedhere
\end{align*}
\end{proof}

We will also need the analogue of the hypercontractivity theorem for matrix-valued
functions. Specifically, we use the following corollary:
\begin{theorem}[{\cite[Lemma~6]{BARdW08}}]
\label{lm:mathyp}
For every $g : \bools^n \rightarrow \Cbb^{2^c \times 2^c}$, $c \in \Nbb_{>0}$, and
$\delta \in \brac{0,1}$ such that $g(x)$ is a density matrix for every $x$, \[
\sum_{U \subseteq \brac{n}} \delta^{\abs{U}} \norm{\wh{g}(U)}_1^2 \le 2^{2\delta c}.
\]
\end{theorem}
By taking $\delta = 1/c$ we immediately obtain \[
\sum_{U \subseteq \brac{n}:\abs{U} = r} \norm{\wh{g}(U)}_1^2 \le 4c^r.
\]

\subsection{Classical Lower Bound}
In this section we prove our classical lower bound. 
\begin{theorem}
\label{thm:classlb}
For any $\alpha \in (0,1)$, $d, k, n \in \Nbb$ with $d \le k | \alpha n$, let $f
: \bools^k \rightarrow \bools$ and $q : \bools^k \rightarrow \brac{0,1}$ be a
function and probability distribution such that $\E[x \sim q]{f \chi_{S}} = 0$
whenever $\abs{S} < d$. Let Alice and Bob's inputs to an $\bhp_{k,n}^\alpha$ be
given by the hard distribution from {\rm\cref{sec:harddist}}. Then there is a
constant $D > 0$, depending only on $f$ and $\alpha$, such that if Alice sends
a classical message of fewer than $\eps^{2/d}n^{1-1/d}/D$ bits the probability
that Bob correctly guesses whether their input was a \yes{} or a \no{} instance
is bounded by \[ 
\frac{1}{2} + \eps.
\]
\end{theorem}
We will use the hard distribution from \cref{sec:harddist}. We may assume
Alice's message is a deterministic function of her input, as for a known
distribution any randomized protocol can be derandomized by choosing the random
seed that maximizes its success probability over the distribution. Let $c$ be
the size of the message Alice sends, and let $m : \bools^n \rightarrow \bool$
be the indicator function on whether a string is consistent with Alice's
message. Note that $m$ is a random variable depending on Alice's input, which
is in turn sampled from $Q$, which depends on $\sigma$.

Let $p$ be Bob's posterior distribution on Alice's input, before he looks at
$y$ (but after seeing $m$ and $\sigma$). Let $F$ be the resulting posterior on
$(f(x^i))_{i = 1}^{\alpha n/k}$. After seeing $y$, Bob's best strategy is to
guess \yes{} if $F(y) \ge F(\overline{y})$ and \no{} otherwise (guessing if
they are the same), which will succeed iff the true value of $(f(x^i))_{i =
1}^{\alpha n/k}$ was more likely than its complement (or 1/2 if they are the
same), which happens with probability
\begin{align*}
\sum_{z \in \bools^{\alpha n/k}} F(z) \left( \mathbf{1}[F(z) > F(\overline{z})] + \frac{1}{2}\mathbf{1}[F(z) = F(\overline{z})]\right) &= \frac{1}{2} + \frac{1}{4}\sum_{z \in \bools^{\alpha
n/k}} \abs{F(z) - F(\overline{z})}\\
&\leq \frac{1}{2} + \norm{F - \Uc}_{\text{TV}}.
\end{align*}
It will therefore suffice to bound the distance of $F$ from uniformity. By
Pinsker's inequality and the chain rule for the KL divergence,
\begin{align*}
 \norm{F - \Uc}_{\text{TV}} &\le \sqrt{\frac{1}{2}D_{\text{KL}}(F || \Uc)}\\
 &\le \sqrt{\frac{1}{2}\sum_{i=1}^{\alpha n/k} \E[(F_j)_{j =
 1}^{i-1}]{D_{\text{KL}}(F_i|_{(F_j)_{j = 1}^{i-1}} || \Uc)}}\\
 &\le \sqrt{\frac{1}{2}\sum_{i=1}^{\alpha n/k} \E[(F_j)_{j =
  1}^{i-1}]{\E[b \sim F_i| (F_j)_{j = 1}^{i-1}]{b}^2}},
\end{align*}
where the last line comes from $D_{\text{KL}}(F_i|_{(F_j)_{j = 1}^{i-1}} || \Uc) \leq D_{\chi^2}(F_i|_{(F_j)_{j = 1}^{i-1}} || \Uc) = \E[b \sim F_i| (F_j)_{j = 1}^{i-1}]{b}^2$, where $D_{\chi^2}(q_1\|q_2):=\sum_x \frac{(q_1(x)-q_2(x))^2}{q_2(x)}$ is the chi-square divergence.
Then, letting $\Fc^i_z$ denote the event that $f(x^j) =
z_j$ for all $j \in \brac{i-1}$, 
\begin{align*}
  \E[\sigma, m]{\norm{F - \Uc}_{\text{TV}}} &\le \sqrt{\frac{1}{2}\sum_{i=1}^{\alpha n/k} \E[(F_j)_{j = 1}^{i-1}]{\E[m,\sigma | (F_j)_{j = 1}^{i-1}]{\E[b \sim F_i| m, \sigma, (F_j)_{j = 1}^{i-1}]{b}^2}}} \tag{Jensen's inequality}\\
    &\le \sqrt{\frac{1}{2}\sum_{i=1}^{\alpha n/k} \max_{z \in \bools^{\alpha
    n/k}}\E[m,\sigma | \Fc_z^i]{\E[b \sim F_i| m, \sigma, \Fc_z^i]{b}^2}}\\
    &= \sqrt{\frac{1}{2}\sum_{i=1}^{\alpha n/k} \max_{z \in \bools^{\alpha
    n/k}}\E[m,\sigma | \Fc_z^i]{{\E[x \sim p | \Fc^i_z]{f(x^i)}^2}}}.
\end{align*}
So it will suffice to prove that, if $c \le \eps^{2/d}n^{1-1/d}/D$, then
\[
\E[\sigma, m | \Fc^i_z]{\E[x \sim p | \Fc^i_z]{f(x^i)}^2 } \le 2\eps^2/n
\quad\text{for all } i \in \brac{\alpha n/k}, z \in \bools^{\alpha n/k}.
\]
This is formally done in the next result.
\begin{lemma}
\label{lm:fbounds}
For all $i \in \brac{\alpha n/k}$ and all $z \in \bools^{i-1}$, let $\Fc_z^i$ be the event that $f(x^j) = z_j$ for every $j \in \brac{i-1}$. Then \[
\E[\sigma, m | \Fc_z^i]{\E[x \sim p | \Fc_z^i]{f(x^i)}^2} \le \bO{\paren*{\frac{c}{{n}}}^d}
\]
where $c$ is the number of bits Alice sends.
\end{lemma}

Given the above result, Bob's advantage over random guessing is bounded by \[
\sqrt{\sum_{i=1}^{\alpha n/k} \bO{\paren*{\frac{c}{n}}^d}}
\]
and so as long as $c \le \eps^{2/d} n^{1 - 1/d}/D$ for some constant $D$, the
advantage will be at most $\eps$, proving \cref{thm:classlb}.

\subsubsection{Proof of \texorpdfstring{\Cref{lm:fbounds}}{Lemma 8}}
We now prove \Cref{lm:fbounds}, which gives the necessary bound on the squared conditional expectation of $f(x^i)$, conditioned on $(f(x^j))_{j = 1}^{i-1}$. Let \[
B_{\overline{i}} = \bigcup_{j \in \brac{\alpha n/k} \setminus \set{i}} B_j
\] and let $\sigma_{\overline{i}} = \sigma|_{B_{\overline{i}}}$. We will start
by bounding \[
 \E[\sigma | m, \sigma_{\overline{i}}, \Fc_z^i]{\E[x \sim p | \Fc_z^i]{f(x^i)}^2}.
\]
Let $S_i = \brac{n} \setminus \sigma(B_{\overline{i}})$ and let $\nu_i : \bools^{S_i} \rightarrow \brac{0,1}$ be given by \[
\nu_i(x) = 2^{i-1}\sum_{y \in \bools^{\sigma(B_{\overline{i}})}} m(x
\cdot y) Q_{\overline{i}}(x \cdot y)\Fc_z^i(x \cdot y),
\]
where we abuse notation and let $\Fc_z^i : \bools^n \rightarrow \bool$ be the indicator function on $x$ causing $\Fc_z^i$ to hold (recall $\Fc_z^i$ is the event
that $f(x^j) = z_j$ for every $j \in \brac{i-1}$) and \[
Q_{\overline{i}}(x) = 
\prod_{j \in \brac{\alpha n/k} \setminus \set{i}}q(x^j),
\]
noting that $Q_{\overline{i}}(x \cdot y)$ depends only on $y$ and is
independent of $\sigma$ conditioned on $\sigma_{\overline{i}}$. This definition
of $\nu_i$ is a valid $\brac{0,1}$-bounded function because $f$ is balanced over
$q$ and therefore
\begin{align*}
\sum_{y \in \bools^{\sigma(B_{\overline{i}})}} Q_{\overline{i}}(x \cdot
y)\Fc_z^i(x\cdot y) = \prod_{j\in[i-1]}\sum_{y\in\{1,-1\}^{k}} q(y^j)\frac{1+f(y^j)z_j}{2} = 2^{1-i} \qquad\forall x \in \bools^{S_i}.
\end{align*}
Note that for any $x\in\bools^{S_i}$ and $y\in\bools^{\sigma(B_{\overline{i}})}$, $(x \cdot y)^i$ depends only
on $x$, so we extend our definition of $x^i$ to $\bools^{S_i}$ by $x^i = (x
\cdot 1^{B_{\overline{i}}})^i$.  Now letting $\Nc = \E[x]{q(x^i)\nu_i(x)}$, we
have
\begin{align*}
\prob[x' \sim p | \Fc_z^i]{x'_{S_i} = x} = \frac{\sum_{t \in
\bools^{\sigma(B_{\overline{i}})}}Q(t\cdot x)m(t \cdot x)\Fc_z^i(t \cdot x)}{\sum_{t \in
\bools^n}Q(t)m(t)\Fc_z^i(t)} = \frac{q(x^i)\nu_i(x)}{2^{\abs{S_i}} \Nc}.
\end{align*}
Now, $\sigma$ is not independent of
$m$, so we will need to rewrite the expression. The
probability that Alice sends message $m$ conditioned on $\Fc_z^i$ and
$\sigma$
is \[ \prob{m | \sigma, \Fc_z^i} = 
\frac{\sum_{x \in \bools^n} Q(x) m(x)
\Fc_z^i(x)}{\sum_{x \in \bools^n} Q(x) \Fc_z^i(x)} = 2^k\Nc
\]
by the definition of $\Nc$ and again using the fact that $f$ is balanced under $q$. So
\begin{align*}
\E[\sigma | m, \sigma_{\overline{i}}, \Fc_z^i]{\E[x \sim p | \Fc_z^i]{f(x^i)}^2} &=
\sum_\tau \prob{\sigma = \tau | m, \sigma_{\overline{i}}, \Fc_z^i} \left.\E[x
\sim p | \Fc_z^i]{f(x^i)}^2\right|_{\sigma = \tau}\\
&= \sum_\tau \prob{\sigma = \tau | \sigma_{\overline{i}}, \Fc_z^i} \frac{\prob{m | \sigma = \tau, \sigma_{\overline{i}}, \Fc_z^i}}{\prob{m | \sigma_{\overline{i}}, \Fc_z^i}}\left.\E[x
\sim p | \Fc_z^i]{f(x^i)}^2\right|_{\sigma = \tau}\\
&= 2^k p_m^{-1} \E[\sigma | \sigma_{\overline{i}}, \Fc_z^i]{\Nc \E[x
\sim p | \Fc_z^i]{f(x^i)}^2},
\end{align*}
where $p_m$ is the probability of Alice sending $m$ conditioned on
$\sigma_{\overline{i}}, \Fc_z^i$ (if this probability is $0$ then we define $\E[x
\sim p | \Fc_z^i]{f(x^i)}^2 = 0$). In order to bound the above right-hand side, we shall employ the following technical result whose proof is deferred to the end of the section.
\begin{lemma}
\label{lm:recbound}
Let $m \ge k$ be integers. Let  $f:\bools^k \rightarrow
\bools$ be a function and $q:\bools^k\to[0,1]$ be a probability distribution such that $\wh{fq}(U) = 0$ for all $U$ with $\abs{U} < d$. Let $\pi:\brac{k} \rightarrow \brac{m}$ be a random injection and let $x^1 = x^\pi_{\brac{k}}$. Let $\nu : \bools^m
\rightarrow \brac{0,1}$ and let $\Qc = \E[x]{\nu(x)}$, $C = -\log \Qc$, $\Nc =
\E[x]{q(x^1)\nu(x)}$.  If $p$ is $0$ when $\Nc = 0$ and the distribution
on $\bools^m$ given by
\[ 
p(x) = \frac{q(x^1)\nu(x)}{2^{m} \Nc},
\]
then \[
\E[\pi]{\Nc \E[x\sim p]{f(x^1)}^2} \le \bO{\Qc\paren*{\frac{C+1}{m}}^d}
\]
where the constant depends on $k$ but nothing else.
\end{lemma}

Associating $\pi$ with $\sigma|_{B_i}$, and
noting that conditioning on $\sigma|_{B_{\overline{i}}}$ and $\Fc_z^i$
leaves $\sigma|_{B_i}$ uniformly random, \cref{lm:recbound} tells us that \[
 \E[\sigma | \sigma_{\overline{i}}, \Fc_z^i]{\Nc\E[x \sim p | \Fc_z^i]{f(x^i)}^2} \le 
 \bO{\Qc\paren*{\frac{C + 1}{n}}^d}
\]
where $\Qc = \E[x]{\nu_i(x)}$ and $C = -\log \Qc$. We can therefore write 
\begin{align*}
\E[\sigma, m | \sigma_{\overline{i}}, \Fc_z^i]{\E[x \sim p | \Fc_z^i]{f(x^i)}^2} &=
2^k\sum_m  \E[\sigma | \sigma_{\overline{i}}, \Fc_z^i]{\Nc\E[x \sim p |
\Fc_z^i]{f(x^i)}^2} 
\le \sum_m \left.\bO{\Qc \paren*{\frac{C + 1}{n}}^d}\right|_m
\end{align*}
and then observe that \[
\sum_m \Qc|_m = 2^{i- 1 - \abs{S_i}} \sum_{x \in \bools^n}\sum_m m(x) Q_{\overline{i}}(x)
\Fc_z^i(x) =  2^{i-1-\abs{S_i}} \sum_{x \in \bools^n} Q_{\overline{i}}(x) \Fc_z^i(x) = 1
\]
as every $x$ has exactly one $m$ such that $m(x) = 1$. So $\mu(m) = \Qc$ is a probability distribution with support $2^c$, and so \[
\sum_m \bO{\Qc \paren*{C + 1}^d} =  \sum_m \mu(m) \bO{\paren*{1 + \log
\frac{1}{\mu(m)}}^d} \le \bO{1 + \sum_m \mu(m) \log^d\frac{1}{\mu(m)}}
\]
and the sum inside the $\bO{\cdot}$ can be bounded by
\begin{align*}
c^d + \sum_{\lambda = 0}^\infty \sum_{\substack{m :\\ \mu(m) \in (2^{-c
-\lambda-1}, 2^{-c - \lambda}\rbrack}} \mu(m)\log^d \frac{1}{\mu(m)} &\le c^d +
\sum_{\lambda = 0}^\infty 2^{-\lambda}(c + \lambda + 1)^d = \bO{c^d}.
\end{align*}
So \[
\E[\sigma, m| \sigma_{\overline{i}}, \Fc_z^i]{\E[x \sim p | \Fc_z^i]{f(x^i)}^2} \le
\bO{\paren*{\frac{c}{n}}^d}
\]
and if $c = 0$, Alice sends no bits and therefore Bob's posterior on $f(x^i)$ is uniform by the independence of different blocks.
\cref{lm:fbounds} then follows by averaging over all $\sigma_{\overline{i}}$.

\subsubsection{Proof of \texorpdfstring{\Cref{lm:recbound}}{Lemma 9}}

We now prove \cref{lm:recbound}. If $C \le k$, we may replace $\nu$ with $\nu' = 2^{-k}\nu$, and a bound in that
setting will give a bound for the original $\nu$ with only constant factor
loss. On the other hand, if $C \ge m$ the bound holds trivially. So without
loss of generality we will assume $C \in \brac{k, m}$. Likewise, if $\Qc = 0$
then $\Nc = 0$ for all $\pi$ and so the lemma follows trivially, so we assume
$\Qc > 0$.

We proceed by induction on $k$. If $k = 0$ or $0 = d < k$, the result holds
trivially. Otherwise we consider two events, depending on whether $\Nc < \Qc /
2^{k+1}$. We start by considering the event where this holds. Start by noting that,
writing $q^1$ for the $m$-bit function with $q^1(x) = q(x^1)$,
\begin{align*}
\Nc &= \wh{q^1\nu}(\emptyset)\\
&= \sum_{U \subseteq \brac{k}} \wh{q}(U)\wh{\nu}(\pi(U))\\
&= \wh{q}(\emptyset)\wh{\nu}(\emptyset) + \sum_{U \subseteq
\brac{k}:\abs{U}>0}  \wh{q}(U)\wh{\nu}(\pi(U))\\
&= 2^{-k}\Qc + \sum_{U \subseteq \brac{k} : \abs{U} > 0}
\wh{q}(U)\wh{\nu}(\pi(U))\\
& \ge  2^{-k}\Qc - \norm{q}_2 \sqrt{\sum_{U \subseteq \brac{k} \setminus
\emptyset} \wh{\nu}(\pi(U))^2} \tag{Cauchy-Schwarz and Parseval's identity}\\
&\ge 2^{-k} \Qc - 2^{-k/2} \sqrt{\sum_{U \subseteq \brac{k} : \abs{U} > 0}
\wh{\nu}(\pi(U))^2} \tag{$q(x)\leq 1$},
\end{align*}
so this case can only hold when at
least one $U$ has $\wh{\nu}(\pi(U))^2 \ge 2^{-2k-2}\Qc^2$.

Let $\Ec_{U,V}$ be the event that $\pi(U) = V$. Whenever $\abs{U} = \abs{V} =
r$, we have \[
\prob{\Ec_{U,V}} = \binom{m}{r}^{-1} \le \paren*{\frac{r}{m}}^r
\]
and for all $r \ge d$ we can take \[
\E[\pi|\Ec_{U,V}]{\Nc \E[x\sim p]{f(x^1)}^2} \le \Qc
\]
trivially as $\Nc \le \Qc$. When $r < d$, we will use our inductive hypothesis
to bound the squared expectation of $f$ conditioned on values of the
$V$-co-ordinates. Then, for any $\tau : U \rightarrow \brac{m}$ such that
$\tau(U) = V$, let $\Tc_{\tau}$ be the event that $\pi|_U = \tau$. We now bound
the squared expectation of $f$ conditioned on $\Tc_{\tau}$. Define $q_U$ as the
marginal of $q$ on the indices in $U$, and for every $z \in \bools^U$ where
$q_U(z) \not = 0$, we define each of the functions from the lemma statement
restricted by the $V$-co-ordinates being $z$. Let $q_z : \bools^{\brac{k} \setminus U}
\rightarrow \brac{0,1}$, $f_z :  \bools^{\brac{k} \setminus U} \rightarrow \bools$,
$\nu_z : \bools^{\brac{m} \setminus V} \rightarrow \brac{0,1}$, $p_z : \bools^{\brac{m} \setminus V} \rightarrow \brac{0,1}$, $\Qc_z, C_z, \Nc_z \in
\Rbb$, as follows:
\begin{align*}
q_z(x) &= \prob[X \sim q | X_U = z]{X_{\brac{k} \setminus U} = x}, \\ 
f_z(x) &= f(x \cdot z), \\
\nu_z(x) &= \nu(x \cdot z'), \\
\Qc_z &= \E[x]{\nu_z(x)}, \\
C_z &= -\log \Qc_z, \\
\Nc_z &= \E[x]{q_z(x^1)\nu_z(x)}, \\
p_z(x) &= \prob[X \sim p| X^\pi_U = z]{X_{\brac{m} \setminus V} = x} =
\frac{q_z(x^1)\nu_z(x)}{2^{m-\abs{U}}\Nc_z},
\end{align*}
where $z'$ is the string in $\bools^V$ with $z_i = z'_{\pi(i)}$ and we extend
our definition of $x^1$ so that for $x \in \bools^{\brac{m} \setminus V}$ it is
$x^{\pi}_{\brac{k} \setminus U}$. Now, for all $W \subseteq \brac{k} \setminus
U$ with $\abs{W} < d - \abs{U}$, \[
\E[x \sim q_z]{f_z(x)\chi_W(x)} = \E[X \sim q | X_U = z]{f(X) \chi_W(X)} =
q_U(z)^{-1}  \E[X \sim q]{f(X) \chi_W(X)\mathbf{1}[X_U = z]} = 0
\]
as the indicator function $\mathbf{1}[X_U^\pi = z]$ has Fourier degree $\le \abs{U}$. Moreover, whenever $q_U(z) > 0$, \[
\prob[x\sim p]{x_U^\pi = z} = \frac{\sum_{\substack{x \in \bools^m:\\ x_U^\pi =z}}
q(x^1) \nu(x)}{\sum_{x \in \bools^m} q(x^1) \nu(x)} = \frac{\sum_{x \in
\bools^{\brac{m} \setminus V}} q_U(z)q_z(x^1)\nu_z(x)}{2^m\Nc} = 2^{-\abs{U}}q_U(z)
\frac{\Nc_z}{\Nc}
\]
with the probability being $0$ when $q_U(z) = 0$. Therefore, as the inductive
hypothesis only depends on the size of the set being injected, 
\begin{align*}
\E[\pi|\Tc_{\tau}]{\Nc\E[x \sim p]{f(x^1)}^2} &=
\E[\pi|\Tc_{\tau}]{\Nc\paren*{\sum_{\substack{z \in \bools^U:\\ q_U(z) > 0}} \E[x \sim
p_z]{f_z(x^1)} \prob[x \sim p]{x_U^\pi = z}}^2}\\ &\le \E[\pi|\Tc_{\tau}]{\sum_{\substack{z
\in \bools^U:\\ q_U(z) > 0}} 2^{-\abs{U}}q_U(z)\Nc_z\E[x
\sim p_z]{f_z(x^1)}^2} \tag{Cauchy-Schwarz}\\
&\le \sum_{\substack{z \in \bools^U:\\ q_U(z) > 0}} \bO{2^{-|U|} q_U(z) \Qc_z
\paren*{\frac{C_z+1}{m - \abs{U}}}^{d - \abs{U}}} \tag{inductive hypothesis}\\
&\le  \sum_{z \in \bools^{\abs{U}}}\bO{\Qc_z
\paren*{\frac{C_z+1}{m}}^{d - \abs{U}}}.
\end{align*}
Now, note that $\sum_{z \in \bools^U} \Qc_z = 2^{\abs{U}}\Qc$, so we can write
\begin{align*}
\sum_{z \in \bools^{\abs{U}}} \Qc_z C_z^{d - \abs{U}} &= 2^{\abs{U}} \Qc
\sum_{z \in
\bools^{\abs{U}}}\frac{\Qc_z}{2^{\abs{U}}\Qc} \paren*{C - \abs{U} -\log
\frac{\Qc_z}{2^{\abs{U}}\Qc}}^{d -\abs{U}}\\
&\le \bO{\Qc} \paren*{\paren*{C - \abs{U}}^{d - \abs{U}} + \bO{1}}\\
&\le \bO{\Qc (C+1)^{d - \abs{U}}}
\end{align*}
since \[
\E[z \sim \mu]{\log^{d - \abs{U}} \frac{1}{\mu(z)}}
\]
is bounded by a constant for any distribution $\mu(z)$ with constant support size. So we have \[
\E[\pi|\Ec_{U,V}]{\Nc\E[x \sim p]{f(x^1)}^2} \le \max_{\tau : U \rightarrow V}
\E[\pi|\Tc_\tau]{\Nc\E[x \sim p]{f(x^1)}^2} \le
\bO{\Qc\paren*{\frac{C+1}{m}}^{d - \abs{U}}}.
\]
Therefore, for all $U, V$ with $\abs{U} = \abs{V} = r < d$, we have \[
\E[\pi|\Ec_{U,V}]{\Nc\E[x \sim p]{f(x^1)}^2}\prob[\pi]{\Ec_{U,V}} \le
\bO{\Qc \paren*{\frac{r}{m}}^r \cdot \paren*{\frac{C+1}{m}}^{d - r}}.
\]
Now for each $r \in \brac{k}$, if we write $\Ec_r$ for the event that
$\Ec_{U,V}$ for any $U \subseteq \brac{k}$ and $V\subseteq[m]$ with $|U| = \abs{V} = r$ and
$\wh{\nu}(V)^2 \ge 2^{-2k-2}\Qc^2$ (recall this holds whenever $\Nc < \Qc /
2^{k+1}$), then by \cref{lm:genkkl} the number of $V$
that can satisfy this is at most
\[
\bO{\paren*{\frac{C+1}{r}}^r}
\]
(where $2^{-2k}$ is absorbed in the big-O notation) and so if $\Ec = \bigcup_{r = 1}^{k} \Ec_r$,
\begin{align*}
\E[\pi|\Ec]{\Nc\E[x \sim p]{f(x^1)}^2}\prob[\pi]{\Ec} &\le \Qc\sum_{r=1}^{d-1}
\binom{k}{r} \cdot \bO{\frac{C+1}{r}^r} \cdot \paren*{\frac{r}{m}}^r \cdot
\bO{\paren*{{\frac{C+1}{m}}}^{d - r}}  \\ &\phantom{\le} +  \Qc\sum_{r=d}^k \binom{k}{r} \cdot
\bO{\frac{C+1}{r}^r}
\cdot \paren*{\frac{r}{m}}^r \cdot 1\\
&\le \bO{\Qc\paren*{{\frac{C+1}{m}}}^d}.
\end{align*}
Then, let
$\Ac$ be the event that $\Nc \ge \Qc /2^{k+1}$. Then \[
\E[\pi]{\Nc\E[x \sim p]{f(x^1)}^2} \le \E[\pi|\Ec]{\Nc\E[x \sim
p]{f(x^1)}^2}\prob[\pi]{\Ec} + \E[\pi|\Ac]{\Nc\E[x \sim
p]{f(x^1)}^2}\prob[\pi]{\Ac} 
\]
so it only remains to bound the second term in this sum. We have 
\begin{align*}
\E[x \sim p]{f(x^1)} &= \frac{1}{2^m \Nc}\sum_{x \in \bools^m} f(x^1) q(x^1) \nu(x)\\
&= \frac{1}{2^m \Nc}\sum_{x \in \bools^m} \sum_{U \subseteq \brac{k}} \wh{fq}(U) \chi_U(x^1) \nu(x)\\
&= \frac{1}{2^m \Nc}\sum_{x \in \bools^m} \sum_{U \subseteq \brac{k}}  \wh{fq}(U) \chi_{\pi(U)}(x) \nu(x) \\
&= \frac{1}{\Nc}  \sum_{U \subseteq \brac{k}}  \wh{fq}(U) \wh{\nu}(\pi(U))
\end{align*}
and so, by Cauchy-Schwarz inequality, Parseval's identity, and $\|fq\| = \bO{1}$,  \[
\Nc\E[x \sim p]{f(x^1)}^2 \le \frac{\bO{1}}{\Nc} \sum_{\substack{U \subseteq
\brac{k}: \abs{U} \ge d}} \wh{\nu}(\pi(U))^2
\]
and finally, by the definition of $\Ac$ and again applying \cref{lm:genkkl},
\begin{align*}
\E[\pi|\Ac]{\Nc\E[x \sim p]{f(x^1)}^2}\prob[\pi]{\Ac} &\le
\bO{1}\E[\pi]{\Qc^{-1} \sum_{\substack{U \subseteq
\brac{k}: \abs{U} \ge d}} \wh{\nu}(\pi(U))^2}\\
&\le \bO{1} \sum_{r = d}^{k} \Qc^{-1}\binom{m}{r}^{-1} \bO{\Qc^2 \paren*{\frac{C+1}{r}}^r}\\
&\le \bO{\Qc\paren*{\frac{C+1}{m}}^d}.
\end{align*}

\subsection{Quantum Lower Bound}
In this section we prove our quantum lower bound.
\begin{theorem}
\label{thm:quantlb}
For any $\alpha \in (0,1)$, $d, k, n \in \Nbb$ with $d \le k | \alpha n$, let $f
: \bools^k \rightarrow \bools$ and $q : \bools^k \rightarrow \brac{0,1}$ be a
function and probability distribution such that $\E[x \sim q]{f \chi_{S}} = 0$
whenever $\abs{S} < d$. Let Alice and Bob's inputs to an $\bhp_{k,n}^\alpha$ be
given by the hard distribution from {\rm\cref{sec:harddist}}. Then there is a
constant $D > 0$, depending only on $f$ and $\alpha$, such that if Alice sends
a quantum message of fewer than $\eps^{2/d}n^{1-2/d}/D$ qubits the probability
that Bob correctly guesses whether their input was a \yes{} or a \no{} instance
is bounded by \[
\frac{1}{2} + \eps.
\]
\end{theorem}
We will again use the hard distribution from \cref{sec:harddist}. Let $c > 0$ be
the number of qubits in Alice's message (if $c = 0$ the theorem is trivial),
and let $m : \bools^n \rightarrow \mathcal{D}(\Cbb^{2^c})$ be the function such
that the density matrix of Alice's message is $m(x)$ when Alice sees $x$. For
all $z$ of $\alpha n/k$ bits or fewer, let $\Fc_z$ be the event that $f(x^i) =
z_i$ for all $i \in \brac{\text{len}(z)}$, and write \[
m_{\sigma, z} = \E[x \sim Q | \sigma, \Fc_z]{m(x)}.
\]
Then on receiving $\sigma$ and $y$, Bob's probability of distinguishing which
case he is in is given by the trace distance between $m_{\sigma,y}$ and
$m_{\sigma, \overline{y}}$. In expectation over the inputs, if Bob guesses
correctly with probability $\frac{1}{2} + \delta$, then
\begin{align*}
\delta &\le \frac{1}{4} \E[\sigma, y]{\norm{m_{\sigma,y} -
m_{\sigma,\overline{y}}}_1}\\
&\le \frac{1}{2}\E[\sigma, y]{\norm{m_{\sigma, y} -
m_{\sigma, \epsilon}}_1}\\
&\le \frac{1}{2}\sum_{i=1}^{\alpha n / k} \E[\sigma, y]{\norm{m_{\sigma,y_{\le i}} -
m_{\sigma,y_{\le i-1}}}_1},
\end{align*}
where $\epsilon$ denotes the empty string. Our objective is to bound
$\E[\sigma, y]{\norm{m_{\sigma,y_{\le i}} -  m_{\sigma,y_{\le i-1}}}_1}$ by $\eps
/ n$, which can be done with the following lemma, which is the quantum analogue of \cref{lm:fbounds}.
\begin{lemma}
\label{lm:qfbounds}
For all $i \in \brac{\alpha n / k}$, \[
\E[\sigma, y]{\norm{m_{\sigma,y_{\le i}}-
m_{\sigma,y_{\le i-1}}}_1} \le \bO{\paren*{\frac{c}{n}}^{d/2}}.
\]
\end{lemma}

Using the above result, we therefore have
\begin{align*}
\delta &\le \sum_{i=1}^{\alpha n/k}\bO{\paren*{\frac{c}{n}}^{d/2}}
\le \bO{\frac{c^{d/2}}{n^{d/2 - 1}}}
\le \eps
\end{align*}
if $c \le n^{1 - 2/d} \eps^{2/d} / D$ for a sufficiently large constant $D$,
proving \cref{thm:quantlb}.

\subsubsection{Proof of \texorpdfstring{\Cref{lm:qfbounds}}{Lemma 11}}
In this section we will prove the quantum version of \cref{lm:fbounds}. We start by observing that for all $z \in \bools^i$, \[
m_{\sigma,z_{\le i}}- m_{\sigma,z_{\le i-1}} = \frac{1}{2}
\paren*{m_{\sigma,z_{\le i}} - m_{\sigma,z_{\le i - 1}\cdot \overline{z_i}}} =
z_i \E[x | \sigma, \Fc_{z_{\le i -1}}]{f(x^i)m(x)}
\]
as $f(x^i)$ is uniform under $Q$ conditioned on $\sigma$ and $(f(x^j))_{j <
i}$. Let
\begin{align*}
B_{\overline{i}} &= \bigcup_{j \in \brac{\alpha n/k} \setminus \set{i}} B_j, \\
S_i &= [n]\setminus\sigma(B_{\overline{i}}), \\
\sigma_{\overline{i}} &= \sigma|_{B_{\overline{i}}},
\end{align*}
and we will proceed to bound \[
\E[\sigma | \sigma_{\overline{i}}, \Fc_z]{\norm*{\E[x \sim Q | \sigma,
\Fc_z]{f(x^i)m(x)}}_1}
\]
for all $z \in \bools^{i-1}$ and values of $\sigma_{\overline{i}}$. Write $Q_{\overline{i}} : \bools^n \rightarrow \brac{0,1}$ and $\nu_i :
\bools^{S_i} \rightarrow \Cbb^{2^c \times 2^c}$ as
\begin{align*}
Q_{\overline{i}}(x) &= \prod_{j \in \brac{\alpha n/k} \setminus \set{i}} q(x^j), \\
\nu_i(u) &= 2^{i-1}\sum_{v \in \bools^{\sigma(B_{\overline{i}})}}
Q_{\overline{i}}(u\cdot v)\Fc_z(u \cdot v) m(u \cdot v),
\end{align*}
where we abuse notation and let $\Fc_{z}(x)$ denote the indicator function on whether $x$ makes $\Fc_z$
hold (this depends on $\sigma$). Note that $\nu_i(u)$ is a density matrix for
each $u$, as each $m(u\cdot v)$ is a density matrix, $Q_{\overline{i}}$ is a
probability distribution on $\bools^{\sigma(B_{\overline{i}})}$, and a string
sampled under $Q_{\overline{i}}$ will have an exactly $2^{1-i}$ probability of
satisfying $\Fc_z$. Noting that for all $u \in \bools^{S_i}$, $(u\cdot v)^i$ only
depends on $u$, we write $u^i = (u \cdot 1^{\brac{n} \setminus S_i})^i$. We also write, for every $U \subseteq \brac{k}$, $U^{\rightarrow i}$ for the
copy of $U$ in
$B_i$, i.e., \[
\set{j + (i - 1)k : j \in U}.
\]
With this notation set, we have
\begin{align*}
\E[x \sim Q | \sigma, \Fc_z]{f(x^i)m(x)} &= 2^{i-1}\sum_{x \in \bools^n} f(x^i)
Q(x) \Fc_z(x)m(x)\\
&= 2^{(\alpha - 1)n}\sum_{u \in \bools^{S_i}} q(u^i) f(u^i) \nu_i(u) \tag{definition of $\nu_i$ and $Q$}\\
&= 2^{(\alpha - 1)n}\sum_{u \in \bools^{S_i}} \sum_{U \subseteq \brac{k}}
\wh{qf}(U) \chi_U(u^i) \nu_i(u)\\
&= 2^k\sum_{U \subseteq \brac{k}}\wh{qf}(U) 2^{-\abs{S_i}}\sum_{u \in \bools^{S_i}}
\chi_{\sigma(U^{\rightarrow i})}(u) \nu_i(u) \tag{$|S_i| = (1-\alpha)n + k$}\\
&= 2^k\sum_{U \subseteq \brac{k}}\wh{qf}(U) \wh{\nu_i}(\sigma(U^{\rightarrow i})).
\end{align*}
Therefore,
\begin{align*}
\norm*{\E[x \sim Q | \sigma, \Fc_z]{f(x^i)m(x)}}_1 &\le 2^k\sum_{U
\subseteq \brac{k}}\abs{\wh{qf}(U)} \norm{\wh{\nu_i}(\sigma(U^{\rightarrow i}))}_1\\
&\le 
\sum_{\substack{U \subseteq B_i:\\ \abs{U} \ge d}}
\norm{\wh{\nu_i}(\sigma(U))}_1 \tag{$\widehat{qf}(U) = 0$ if $|U|<d$ and $\widehat{qf}(U) \leq 2^{-k}$ otherwise}\\
&\le 2^{k/2} \sqrt{\sum_{\substack{U \subseteq B_i:\\ \abs{U} \ge d}}
\norm{\wh{\nu_i}(\sigma(U))}_1^2}. \tag{Cauchy-Schwarz}
\end{align*}
Therefore, as $\sigma|_{B_i}$ is a uniform injection to $S_i$
conditioned on $\sigma_{\overline{i}}$ and $\Fc_z$, 
\begin{align*}
\E[\sigma | \sigma_{\overline{i}}, \Fc_z]{\norm*{\E[x \sim Q | \sigma,
\Fc_z]{f(x^i)m(x)}}_1} &\le 2^{k/2} \sqrt{\sum_{\substack{U \subseteq B_i:\\
\abs{U} \ge d}} \sum_{V \subseteq S_i} \norm{\wh{\nu_i}(V)}_1^2 \prob[\sigma
| \sigma_{\overline{i}}, \Fc_z]{\sigma(U) = V}} \tag{Jensen's inequality}\\
&\le 2^k\sqrt{\sum_{r = d}^k \binom{\abs{S_i}}{r}^{-1} \sum_{V \subseteq S_i: \abs{S_i} = r}
\norm{\wh{\nu_i}(V)}_1^2} \\
&\le 2^k\sqrt{\sum_{r=d}^k \paren*{\frac{r}{(1 - \alpha)n}}^r 4c^r} \tag{\Cref{lm:mathyp}}\\
&= \bO{\paren*{\frac{c}{n}}^{d/2}}
\end{align*}
whenever $c \le n$, proving \cref{lm:qfbounds}.

\section{Upper Bounds}
\subsection{Classical Upper Bound}

\begin{theorem}
    \label{thm:classical-upper-bound}
    For every $\alpha \in (0,1)$, $d, k, n \in \Nbb$ with $d \le k | \alpha n$, let 
    $f: \bools^k \rightarrow \bools$ be a function with sign degree $d$. Let $p$ be a degree-$d$ polynomial sign-representing
    $f$ such that $\max_{z\in\{1,-1\}^k}\abs{p(z)} \leq 1$ and let  $\beta:=\min_{z\in\bools^k}\abs{p(z)}>0$
    be its margin.
    Given \(\eps>0\), $\bhp_{k,n}^\alpha$ can be solved by a randomized protocol with error $\eps$ using
    \begin{align*}
        \mathrm{O}\left(\left(\frac{k\hat{\|}p\hat{\|}_1^2}{\alpha\beta^2}\right)^{1/d}n^{1 - 1/d}\log\frac{1}{\eps}\right) \quad\text{bits of communication}.
    \end{align*}
\end{theorem}
\begin{proof}
Write the Fourier expansion of $p$ as
\[
    p(z)
    =
    \sum_{T\subseteq[k]: |T|\le d}
    \widehat{p}(T)\chi_T(z)
\]
Let Bob's input consist of $m=\alpha n/k$ pairwise disjoint blocks
$B_1,\ldots,B_m$ together with labels \(y_1,\ldots,y_m\in\bools\).
Under the problem's promise, there is some
\(b\in\bools\) such that
\[
    y_i f(x^i)=b
    \qquad
    \text{for every }i\in[m].
\]
Define
\[
    p_y(x)
    :=
    \frac{1}{m}
    \sum_{i=1}^m
    y_i p(x^i) = \frac{1}{m}\sum_{i=1}^m y_i p(x_{B_i}).
\]
Every term in this average has sign $b$ and magnitude at least
$\beta$, meaning that $bp_y(x)\geq \beta$ for all $x\in\{1,-1\}^n$.
It is thus enough for Bob to estimate $p_y(x)$ to additive error
less than $\beta$.

\paragraph{The protocol.}
Let \(C>0\) be a constant to be chosen below.
Note that the case when $C n^{-1/d} \geq 1$ is trivial: $C n^{1-1/d} \geq n$ and thus Alice can send her whole input to Bob using $C n^{1-1/d}$ bits. Thus we assume that $C n^{-1/d} < 1$.

Using public randomness, Alice includes each coordinate \(j\in[n]\)
independently in a sample \(S\) with probability $q := C n^{-1/d}$. Let $L:=2nq$. If $|S| > L$, the sampling is declared to have failed, and Bob outputs an arbitrary bit. Otherwise, Alice sends the values $\{x_j:j\in S\}$ in the publicly known canonical order of $S$ using $|S| \leq 2C n^{1-1/d}$ bits.

\paragraph{Sample is too big.}
Let us bound the probability of error caused by $S$ being too big. The random variable $|S|$ is distributed as
\[
    |S|\sim\operatorname{Bin}(n,q) \qquad\text{with mean}\qquad \mu=qn=Cn^{1-1/d}.
\]
By a Chernoff bound,
\begin{equation}
    \label{eq:abort-error}
    \Pr[S\text{ is too big}] = \Pr[|S|>2\mu]
    \le
    e^{-\mu/3}.
\end{equation}

\paragraph{The estimator.}
For $U\subseteq[n]$, let $I_U:=\mathbf 1[U\subseteq S]$. For \(T\subseteq[k]\) and block $B_i = (b_{i,j})_{j\in[k]}$, define $B_i(T):=\{b_{i,j}:j\in T\}$. Bob computes
\[
    \widetilde{p}_i
    :=
    \sum_{T\subseteq[k]: |T|\le d}
    \widehat p(T)
    \chi_T(x_{B_i})
    \frac{I_{B_i(T)}}{(C n^{-1/d})^{|T|}},
\]
where the term is computable whenever
$I_{B_i(T)}=1$, since all coordinates involved have been revealed.
We use the convention that the $T=\emptyset$ term is simply
\(\widehat{p}(\emptyset)\). Notice that $\E[S]{I_U}=q^{|U|} = (C n^{-1/d})^{|U|}$ for every \(U\subseteq[n]\), consequently $\E[S]{\widetilde{p}_i} = p(x_{B_i})$.

Finally, Bob computes
\[
    \widetilde{p}_y
    :=
    \frac{1}{m}
    \sum_{i=1}^m y_i\widetilde{p}_i,
\]
which is an unbiased estimator for $p_y(x)$, i.e., $\E[S]{\widetilde{p}_y} = p_y(x)$, and outputs $\operatorname{sgn}(\widetilde{p}_y)$.

\paragraph{Variance.}
Because Bob's blocks are pairwise disjoint and the coordinates are
sampled independently, the random variables $\widetilde{p}_1,\ldots,\widetilde{p}_m$ are mutually independent.  Therefore
\[
    \Var[S]{\widetilde{p}_y}
    =
    \frac{1}{m^2}
    \sum_{i=1}^m
    \Var[S]{\widetilde{p}_i}
    \leq
    \frac{1}{m^2}
    \sum_{i=1}^m
    \E[S]{\widetilde{p}_i^2}.
\]

Fix a block \(i\in[\alpha n/k]\).  For \(T,U\subseteq[k]\),
\[
    I_{B_i(T)}I_{B_i(U)}
    =
    I_{B_i(T\cup U)},
\]
and therefore
\[
    \E[S]
    {
        \frac{
            I_{B_i(T)}I_{B_i(U)}
        }{
            q^{|T|+|U|}
        }
    }
    =
    q^{-|T\cap U|}.
\]
It follows that
\begin{align*}
	\E[S]{\widetilde{p}_i^2}
    &=
    \sum_{T,U\subseteq[k]: |T|,|U|\le d}
    \widehat{p}(T)\widehat{p}(U)
    \chi_T(x_{B_i})
    \chi_U(x_{B_i})
    q^{-|T\cap U|} 
    \leq
    \sum_{T,U\subseteq[k]: |T|,|U|\le d}
    \abs{\widehat{p}(T)}\,\abs{\widehat{p}(U)}
    q^{-|T\cap U|}.
\end{align*}
Since $|T\cap U|\leq d$ and \(q<1\), $q^{-|T\cap U|} \leq q^{-d}$. Thus $\E[S]{\widetilde{p}_i^2} \le \hat{\|}p \hat{\|}_1^2q^{-d}$ and hence
\begin{align*}
	\Var[S]{\widetilde{p}_y}
    \le
    \frac{\hat{\|}p \hat{\|}_1^2}{mq^d}
    =
    \frac{k\hat{\|}p \hat{\|}_1^2}{\alpha n q^d}
    =
    \frac{k\hat{\|}p \hat{\|}_1^2}{\alpha C^d}.
\end{align*}
An incorrect output in an execution requires
\[
    |\widetilde{p}_y-p_y(x)|\geq\beta.
\]
Chebyshev's inequality therefore gives
\begin{equation}
    \label{eq:estimation-error}
    \prob[S]{
        \operatorname{sgn}(\widetilde{p}_y)\neq b
    }
    \le
    \frac{k\hat{\|}p \hat{\|}_1^2}{\alpha\beta^2C^d}.
\end{equation}
By taking \[C \geq \max\left\{\left(\frac{6k\hat{\|}p \hat{\|}_1^2}{\alpha \beta^2}\right)^{1/d}, \frac{3}{n^{1-1/d}}\ln(6)\right\},\]
the total error probability according to \eqref{eq:abort-error} and
\eqref{eq:estimation-error} is at most
\(1/3\). Finally, by repeating the protocol $O(\log(1/\eps))$ times and taking the majority vote of the outputs, the error probability can be reduced to at most \(\eps\).
Thus, we have designed a one-way protocol with error
at most $\eps$ and randomized communication complexity
\[
    \mathrm{O}_{f,\alpha,\varepsilon}
    \left(n^{1-1/d}\right). \qedhere
\]

\end{proof}

\subsection{Quantum Upper Bound}
\label{sec:quantum-upper-bound}

We next give a quantum upper bound matching the previously proven lower bound
when the sign degree is even. Our proof closely follows the quantum upper-bound proof for Boolean Hidden Hypermatching due to Shi, Wu, and Yu
\cite[Appendix~A.1]{SWY12}.  In their protocol, one copy of a
fingerprint state reveals the parity of a uniformly random pair, and
the generalized birthday paradox implies that
$O(n^{1-2/k})$ copies suffice to recover the $k/2$ pair parities
belonging to one common $k$-set, for even $k$. We apply the same argument to the pairs comprising a sampled Fourier
monomial of size at most $d$.

\begin{theorem}
    \label{thm:quantum-upper-bound}
    For every $\alpha \in (0,1)$, $d, k, n \in \Nbb$ with $d \le k | \alpha n$, let 
    $f: \bools^k \rightarrow \bools$ be a function with sign degree $d \geq 1.$\footnote{$d\leq 2$ has been resolved in \cite{DoriguelloMontanaro20}.}       
    Let $p$ be a degree-$d$ polynomial sign-representing
    $f$ such that $\max_{z\in\{1,-1\}^k}\abs{p(z)} \leq 1$ and let  $\beta:=\min_{z\in\bools^k}\abs{p(z)}$
    be its margin.  Then, for every $\eps>0$, there exists a one-way quantum communication protocol solving $\bhp_{k,n}^\alpha$ with error at most $\eps$ sending 
    \[
        \mathrm{O}
    \left(
        \frac{\hat{\|}p\hat{\|}_1^2}{\beta^2}\left(\frac{k}{\alpha}\right)^{1/\lceil d/2 \rceil} n^{1-1/\lceil d/2 \rceil}\log{n}\log\frac1\eps
    \right) \quad\text{qubits}.
    \]
    In particular, for constant $\alpha$, $k$, and $\eps$, the communication complexity is $\mathrm{O}(n^{1-1/\lceil d/2 \rceil} \log{n})$. 
\end{theorem}
\begin{proof}
Write the Fourier expansion of $p$ as
\[
    p(z)
    =
    \sum_{T\subseteq[k]:|T|\le d}
    \widehat p(T)\chi_T(z).
\]
Let $m:= \alpha n/k$ denote the number of blocks relevant to Bob and let $r:=\lceil d/2\rceil$. 

\paragraph{Alice's fingerprint state.}
For every block index $i\in[m]$, introduce $2r$ orthogonal dummy
basis states
\[
    \ket{D_{i,1}},\ldots,\ket{D_{i,2r}}.
\]
The dummy states contain no information about $x$ and will allow every Fourier monomial of degree at most $2r$ to be written as a product of exactly $r$ pair-products. 

Let $N:=n+2rm$. Alice prepares copies of
\begin{equation}
    \label{eq:quantum-fingerprint-state}
    \ket{\psi_x}
    :=
    \frac{1}{\sqrt{N}}
    \left(
        \sum_{j=1}^n x_j\ket{j}
        +
        \sum_{i=1}^m
        \sum_{\ell=1}^{2r}
        \ket{D_{i,\ell}}
    \right).
\end{equation}
Each copy requires $\bO{\log{n}}$ qubits.

\paragraph{Sampling one Fourier monomial.}

A trial begins with Bob sampling a set $S\subseteq[k]$ according~to
\begin{equation*}
    \Pr[S]
    =
    \frac{\abs{\widehat{p}(S)}}{\hat{\|}p \hat{\|}_1}.
\end{equation*}
For every $z\in\bools^k$,
\begin{align}
	\E[S]{\hat{\|}p \hat{\|}_1\,
        \operatorname{sgn}(\widehat{p}(S))
        \chi_S(z)}
    =
    \sum_{S\subseteq[k]}
    \frac{\abs{\widehat p(S)}}{\hat{\|}p \hat{\|}_1}
    \hat{\|}p \hat{\|}_1\operatorname{sgn}(\widehat{p}(S))
    \chi_S(z)
    =
    \sum_{S\subseteq[k]}
    \widehat p(S)\chi_S(z)
    =
    p(z).
    \label{eq:sampled-monomial-unbiased}
\end{align}
Thus it suffices to recover the value of one randomly sampled
Fourier monomial on one of Bob's blocks.

\paragraph{Padding and pairing the monomial.}

Since $|S| \leq d$, pad $S$ with $2r-|S|$
distinct dummy positions.
The corresponding dummy variables all have value $+1$, so this does not change the
value of the monomial. Fix a deterministic rule which partitions the resulting $2r$
positions into $r$ pairs $e_1(S),\dots,e_r(S)$. To describe the corresponding measurement, for
$i\in[m]$ and $a\in[k+2r]$ define
\[
    z_{i,a}
    :=
    \begin{cases}
        x_a^i,
            & a\in[k],\\[1mm]
        1,
            & a\in[k+1,k+2r],
    \end{cases}
\]
and
\[
    \ket{v_{i,a}}
    :=
    \begin{cases}
        \ket{j},
            & j = \sigma((i-1)k+a), a\in[k],\\[1mm]
        \ket{D_{i,\ell}},
            & a=k+\ell, \ell\in[2r].
    \end{cases}
\]
If $e_\ell(S)=\{a_\ell,b_\ell\}$, define $u_{i,\ell} := z_{i,a_\ell}z_{i,b_\ell}$. By construction,
\begin{equation}
    \label{eq:monomial-pair-factorization}
    \chi_S(x^{i})
    =
    \prod_{\ell=1}^r u_{i,\ell}.
\end{equation}

\paragraph{One copy reveals one random pair-product.}

Fix the sampled Fourier monomial $S$ and its pairing
$e_1(S),\ldots,e_r(S)$. Bob measures one copy of $\ket{\psi_x}$ in a basis containing,
for every $i\in[m]$ and $\ell\in[r]$, the two vectors
\begin{equation}
    \label{eq:pair-measurement-basis}
    \frac{
        \ket{v_{i,a_\ell}}
        \pm
        \ket{v_{i,b_\ell}}
    }{\sqrt{2}}.
\end{equation}
The remaining basis vectors may be chosen arbitrarily to complete an
orthonormal basis. The vectors in \eqref{eq:pair-measurement-basis} are mutually
orthogonal since Bob's blocks are disjoint, the dummy coordinates are
block-specific, and the pairs $e_1(S),\ldots,e_r(S)$
are disjoint.

The probability that the measurement lands in the two-dimensional subspace indexed by $(i,\ell)\in[m]\times[r]$ is $2/N$. Conditioned on this event, the post-measurement state is
\[
    \frac{
        z_{i,a_\ell}\ket{v_{i,a_\ell}}
        +
        z_{i,b_\ell}\ket{v_{i,b_\ell}}
    }{\sqrt{2}}.
\]
Hence the sign of the outcome in
\eqref{eq:pair-measurement-basis} reveals exactly $z_{i,a_\ell}z_{i,b_\ell} = u_{i,\ell}$ and a useful measurement outcome gives Bob $(i,\ell,u_{i,\ell})$ where $(i,\ell)$ is uniform on $[m]\times[r]$. The total probability of obtaining a useful outcome is
\begin{equation*}
    \lambda
    :=
    \frac{2mr}{N}
    =
    \frac{\alpha 2\lceil d/2\rceil}{k+\alpha 2\lceil d/2\rceil} = \bO{\frac{\alpha d}{k}}.
\end{equation*}

\paragraph{The birthday collision.}
The idea is for Alice to send several copies of $\ket{\psi_x}$ to Bob, who measures each copy independently. With constant probability, some block $J\in[m]$ is sampled at least $r$ times given $\mathrm{O}(m^{1-1/r})$ samples according to the generalized birthday paradox, and the corresponding pair labels $\ell\in[r]$ are all distinct with probability $r!/r^r = \Omega_r(1)$ (since conditioned on a sample landing in block $J$, its pair label $\ell\in[r]$
is uniform). 

We can make this reasoning more precise by ``poissonizing'' the number of copies and measurements we use (choosing the number from a Poisson distribution with a matching expectation), making our samples arrive as a Poisson process with rate $\lambda$. Each pair $(i,\ell)$ thus arrives at rate $\lambda / mr = 2/N$. All pair labels of block $J\in[m]$ arrive after $s$ samples with probability $(1-e^{-2s/N})^r$. The probability that no block receives $r$ distinct pair labels after $s$ samples is then (since all blocks are independent)
\begin{align*}
    (1-(1-e^{-2s/N})^r)^m,
\end{align*}
which is at most a small constant, say $1/3$, for which it suffices to set
\begin{align*}
s =\bO{N m^{-1/r}} = \bO{(k/\alpha)^{1/r} n^{1-1/r}}.
\end{align*}
Consequently, if Alice sends (in expectation over her Poisson draw) $c (k/\alpha)^{1/r} n^{1-1/r}$ copies of $|\psi_x\rangle$ for sufficiently large constant $c$, Bob obtains a block $J$
for which he knows all $r$ pair-products $u_{J,1},\ldots,u_{J,r}$ with
probability at least $2/3$. Moreover, this holds even if Alice aborts whenever
she sends more than $100$ times as many states as she would in expectation, so
from hereon we assume she never sends more than $\bO{(k/\alpha)^{1/r}
n^{1-1/r}}$ copies.

On this event, \eqref{eq:monomial-pair-factorization} gives
\begin{equation*}
    \chi_S(x^J)
    =
    \prod_{\ell=1}^r u_{J,\ell}.
\end{equation*}

\paragraph{From one monomial to the hidden bit.}

If the trial succeeds, Bob defines
\begin{equation*}
    Z
    :=
    \hat{\|}p\hat{\|}_1 
    \operatorname{sgn}(\widehat{p}(S))
    y_J
    \chi_S(x^J).
\end{equation*}
If the birthday collision fails, he sets $Z:=0$. In either case, $|Z| \leq \hat{\|}p\hat{\|}_1$.

The distribution of the indices $(i,\ell)\in[m]\times[r]$ produced by the measurement
does not depend on the sampled Fourier monomial $S$: every pair
subspace has the same probability $2/N$.
We choose $J$ using only the observed indices, and not the observed
pair-product values. Therefore, the success event and the selected
block $J$ are independent of the random choice of $S$.

Conditioned on the birthday collision success and on $J$, using
\eqref{eq:sampled-monomial-unbiased},
\begin{align*}
	\E[S]{Z \mid J, \mathrm{success}}
    =
    y_J
    \E[S]{\hat{\|}p\hat{\|}_1\,
        \operatorname{sgn}(\widehat{p}(S))
        \chi_S(x^J)}
    =
    y_Jp(x^J).
\end{align*}
By the promise of $\bhp_{k,n}^\alpha$, $y_Jf(x^J)=b$. Since $p$ sign-represents $f$ with margin $\beta$,
\[
    f(x^J) p(x^J) = \abs{p(x^J)} \geq \beta.
\]
Hence
\[
    b\,
    \E[S]{Z\mid J,\mathrm{success}} = b\,y_Jp(x^J) = f(x^J) p(x^J) \geq \beta.
\]
Since the trial succeeds with probability at least $2/3$ and
$Z=0$ on failure,
\begin{equation}
    \label{eq:quantum-single-trial-bias}
    b\, \E[]{Z} \geq \frac{2\beta}{3}.
\end{equation}

\paragraph{Amplification.}

Repeat the trial independently $R = \Big\lceil\frac{9\hat{\|}p\hat{\|}_1^2}{2\beta^2}\ln{\frac{1}{\eps}}\Big\rceil$
times, using fresh copies of Alice's state. Let $Z_1,\ldots,Z_R$ be the resulting random variables and let Bob output
\[
    \mathrm{sgn}\!
    \left(
        \sum_{t=1}^R Z_t
    \right).
\]
By \eqref{eq:quantum-single-trial-bias}, $b\, \E[]{Z_t} \geq \frac{2\beta}{3}$, while $|Z_t| \leq \hat{\|}p\hat{\|}_1$. Hoeffding's inequality therefore gives
\[
    \Pr\left[
        \mathrm{sgn}\!
        \left(
            \sum_{t=1}^R Z_t
        \right)
        \neq b
    \right]
    \le
    \exp\left(
        -\frac{
            2R\beta^2
        }{
            9\hat{\|}p\hat{\|}_1^2
        }
    \right),
\]
which is at most $\eps$ for the stated choice of $R$.

\paragraph{Communication complexity.}

Each trial uses $\mathrm{O}((k/\alpha)^{1/r} n^{1-1/r})$ copies of the fingerprint state, and every copy contains $\mathrm{O}(\log n)$ qubits. 
Therefore the total communication is
\[
    \mathrm{O}
    \left(
        \frac{\hat{\|}p\hat{\|}_1^2}{\beta^2}\left(\frac{k}{\alpha}\right)^{1/\lceil d/2 \rceil} n^{1-1/\lceil d/2 \rceil}\log{n}\log\frac1\eps
    \right). \qedhere
\]
\end{proof}

\section{Acknowledgements}
This material is based upon work supported by the U.S.\ Department of Energy,
Office of Science, Accelerated Research in Quantum Computing, Fundamental
Algorithmic Research toward Quantum Utility (FAR-Qu). 

This work was performed, in part, at the Center for Integrated
Nanotechnologies, an Office of Science User Facility operated for the U.S.
Department of Energy (DOE) Office of Science by Los Alamos National Laboratory
(Contract 89233218CNA000001) and Sandia National Laboratories (Contract
DE-NA-0003525).

JFD is supported by the Lend\"ulet ``Momentum'' program of the Hungarian
Academy of Sciences under grant agreement no.\ LP2025-8/2025. MB was supported
in part by NSF grant CNS-2046425. NV is supported in part by the French PEPR integrated projects EPIQ (ANR-22-PETQ-0007) and HQI (ANR-22-PNCQ-0002). Part of this work was done while NV was at Boston University, supported in part by NSF grant CNS-2046425.

\section{AI Disclosure}
The problem, hard distribution, and initial proof strategy were developed by
the authors. However, two key ideas came from GPT-6 Astra. The first was the
idea of replacing the traditional Fourier-analytic approach of bounding every
odd parity of the block outputs $f(x^i)$ with analyzing individual blocks
conditional on earlier blocks. The second, deeper, idea was the recursive
argument used in the classical case to deal with the quadratic dependence on
$\Qc/\Nc$.

The proofs and other text of the paper were written by humans. AI tools were
used to help find typos, bugs, and other errors.

\newpage
\bibliographystyle{alpha}
\bibliography{refs}

\newcommand{\etalchar}[1]{$^{#1}$}
\begin{thebibliography}{BCWdW01}

\bibitem[BARdW08]{BARdW08}
Avraham Ben-Aroya, Oded Regev, and Ronald de~Wolf.
\newblock A hypercontractive inequality for matrix-valued functions with
  applications to quantum computing and {LDCs}.
\newblock In {\em Proceedings of the 49th Annual IEEE Symposium on Foundations
  of Computer Science}, 2008.

\bibitem[BCWdW01]{Buhrman2001quantum}
Harry Buhrman, Richard Cleve, John Watrous, and Ronald de~Wolf.
\newblock Quantum fingerprinting.
\newblock {\em Phys. Rev. Lett.}, 87:167902, Sep 2001.

\bibitem[Bon70]{bonami1970etude}
Aline Bonami.
\newblock {\'E}tude des coefficients de {Fourier} des fonctions de
  ${L}^p({G})$.
\newblock {\em Annales de l'Institut Fourier}, 20(2):335--402, 1970.

\bibitem[BT15]{bun2015dual}
Mark Bun and Justin Thaler.
\newblock Dual lower bounds for approximate degree and {M}arkov-{B}ernstein
  inequalities.
\newblock {\em Information and Computation}, 243:2--25, 2015.
\newblock 40th International Colloquium on Automata, Languages and Programming
  (ICALP 2013).

\bibitem[BYJK04]{bar2004exponential}
Ziv Bar-Yossef, T.~S. Jayram, and Iordanis Kerenidis.
\newblock Exponential separation of quantum and classical one-way communication
  complexity.
\newblock In {\em Proceedings of the Thirty-Sixth Annual ACM Symposium on
  Theory of Computing}, STOC '04, page 128–137, New York, NY, USA, 2004.
  Association for Computing Machinery.

\bibitem[CGSV24]{CGSV24}
Chi-Ning Chou, Alexander Golovnev, Madhu Sudan, and Santhoshini Velusamy.
\newblock Sketching approximability of all finite csps.
\newblock {\em Journal of the ACM}, 71(2):1--74, 2024.

\bibitem[DM19]{doriguello2019quantum}
Jo{\~a}o~Fernando Doriguello and Ashley Montanaro.
\newblock Quantum sketching protocols for {H}amming distance and beyond.
\newblock {\em Phys. Rev. A}, 99:062331, Jun 2019.

\bibitem[DM20]{DoriguelloMontanaro20}
Jo{\~a}o~F. Doriguello and Ashley Montanaro.
\newblock Exponential quantum communication reductions from generalizations of
  the boolean hidden matching problem.
\newblock In {\em 15th Conference on the Theory of Quantum Computation,
  Communication and Cryptography (TQC 2020)}, volume 158 of {\em Leibniz
  International Proceedings in Informatics (LIPIcs)}, pages 1:1--1:16, 2020.

\bibitem[GKK{\etalchar{+}}08]{GKKRW07}
Dmitry Gavinsky, Julia Kempe, Iordanis Kerenidis, Ran Raz, and Ronald de~Wolf.
\newblock Exponential separation for one-way quantum communication complexity,
  with applications to cryptography.
\newblock {\em {SIAM} J. Comput.}, 38(5):1695--1708, 2008.

\bibitem[KKL88]{KKL88}
J.~Kahn, G.~Kalai, and N.~Linial.
\newblock The influence of variables on boolean functions.
\newblock In {\em Proceedings of the 29th Annual Symposium on Foundations of
  Computer Science}, SFCS '88, pages 68--80, Washington, DC, USA, 1988. IEEE
  Computer Society.

\bibitem[KKP18]{KKP18}
John Kallaugher, Michael Kapralov, and Eric Price.
\newblock The sketching complexity of graph and hypergraph counting.
\newblock In {\em 2018 IEEE 59th Annual Symposium on Foundations of Computer
  Science (FOCS)}, pages 556--567. IEEE, 2018.

\bibitem[KKS15]{KKS15}
Michael Kapralov, Sanjeev Khanna, and Madhu Sudan.
\newblock Streaming lower bounds for approximating {MAX-CUT}.
\newblock In {\em 26th ACM-SIAM Symposium on Discrete Algorithms (SODA)}, 2015.

\bibitem[Mon11]{montanaro2011new}
Ashley Montanaro.
\newblock A new exponential separation between quantum and classical one-way
  communication complexity.
\newblock {\em Quantum Info. Comput.}, 11(7–8):574–591, July 2011.

\bibitem[Raz99]{raz1999exponential}
Ran Raz.
\newblock Exponential separation of quantum and classical communication
  complexity.
\newblock In {\em Proceedings of the Thirty-First Annual ACM Symposium on
  Theory of Computing}, STOC '99, page 358–367, New York, NY, USA, 1999.
  Association for Computing Machinery.

\bibitem[She07]{sherstov2007separating}
Alexander~A. Sherstov.
\newblock Separating ac0 from depth-2 majority circuits.
\newblock In {\em Proceedings of the Thirty-Ninth Annual ACM Symposium on
  Theory of Computing}, STOC '07, page 294–301, New York, NY, USA, 2007.
  Association for Computing Machinery.

\bibitem[SWY12]{SWY12}
Yaoyun Shi, Xiaodi Wu, and Wei Yu.
\newblock Limits of quantum one-way communication by matrix hypercontractive
  inequality, 2012.
\newblock Manuscript.

\bibitem[VY11]{verbin2011streaming}
Elad Verbin and Wei Yu.
\newblock The streaming complexity of cycle counting, sorting by reversals, and
  other problems.
\newblock In {\em Proceedings of the 2011 Annual ACM-SIAM Symposium on Discrete
  Algorithms (SODA)}, pages 11--25, 2011.

\bibitem[Yao79]{yao1979some}
Andrew Chi-Chih Yao.
\newblock Some complexity questions related to distributive
  computing(preliminary report).
\newblock In {\em Proceedings of the Eleventh Annual ACM Symposium on Theory of
  Computing}, STOC '79, page 209–213, New York, NY, USA, 1979. Association
  for Computing Machinery.

\bibitem[Yao93]{yao1993quantum}
A.~Chi-Chih Yao.
\newblock Quantum circuit complexity.
\newblock In {\em Proceedings of 1993 IEEE 34th Annual Foundations of Computer
  Science}, pages 352--361, 1993.

\end{thebibliography}

\end{document}